\documentclass[a4paper,11pt]{article}

\usepackage[a4paper,tmargin=3truecm,bmargin=3truecm,rmargin=2.5truecm,
lmargin=2.5truecm,twoside,verbose=true]{geometry}

\usepackage{cancel,graphicx}
\usepackage{amsmath,amssymb}
\usepackage[amsmath, hyperref, thmmarks]{ntheorem}
\usepackage[all]{xypic}

\usepackage[pdftex]{hyperref}

\usepackage[english]{babel}

\numberwithin{equation}{section}

\allowdisplaybreaks[1]

\newcommand\caA{{\mathcal A}}

\newcommand\caG{{\mathcal G}}
\newcommand\caI{{\mathcal I}}

\newcommand\caM{{\mathcal M_\theta}}
\newcommand\caS{{\mathcal S}}
\newcommand\caU{{\mathcal U}}
\newcommand\caZ{{\mathcal Z}}
\newcommand\wx{{\widetilde x}}
\newcommand\gone{{ \mathchoice {1\mskip-4mu\mathrm{l} } {1\mskip-4mu\mathrm{l} }{1\mskip-4.5mu\mathrm{l} } {1\mskip-5mu\mathrm{l}} }}
\newcommand\gR{{\mathbb R}}
\newcommand\gK{{\mathbb K}}
\newcommand\gC{{\mathbb C}}
\newcommand\gN{{\mathbb N}}
\newcommand\gZ{{\mathbb Z}}
\newcommand\Omr{\underline{\Omega}_\varepsilon}
\newcommand\algzero{{\mathsf 0}}
\newcommand\algA{{\mathbf A}}
\newcommand\algrA{{\mathbf A}^\bullet}
\newcommand\algB{{\mathbf B}}
\newcommand\algrB{{\mathbf B}^\bullet}
\newcommand\modM{{\boldsymbol M}}
\newcommand\modrM{{\boldsymbol M}^\bullet}
\newcommand\kg{{\mathfrak g}}

\newcommand\ksl{{\mathfrak{sl}}}
\newcommand\kX{{\mathfrak X}}
\newcommand\kY{{\mathfrak Y}}
\newcommand\kS{{\mathfrak S}}
\newcommand\eps{{\varepsilon}}
\newcommand\ad{{\text{\textup{ad}}}}

\newcommand{\grast}{\bullet}
\newcommand\fois{\mathord{\cdot}}
\DeclareMathOperator{\tr}{Tr} % trace
\DeclareMathOperator{\Hom}{\mathsf{Hom}}
\DeclareMathOperator{\Aut}{\mathsf{Aut}}
\newcommand\Matr{{\mathcal{M}}}
\newcommand\Der{{\text{\textup{Der}}}}
\newcommand\Int{{\text{\textup{Inn}}}}
\newcommand\Out{{\text{\textup{Out}}}}
\newcommand\dd{{\text{\textup{d}}}}

\newcommand{\omi}[1]{\buildrel { \buildrel{#1}\over{\vee} } \over .}
\newcommand\Lie{{\text{\textup{Lie}}}}
\newcommand\symes{{\mathchoice{\textstyle\mathsf{S}}{\textstyle\mathsf{S}}%
{\scriptstyle\mathsf{S}}{\scriptscriptstyle\mathsf{S}}}} % signe S en grand, sans-sérif, pour les algèbres
\newcommand\exter{{\textstyle\bigwedge}} % signe wedge en grand, pour les algèbres
\newcommand\Supp{{\text{\textup{Supp}}}}
\theoremsymbol{}
\theorembodyfont{\slshape}
\theoremheaderfont{\normalfont\bfseries}
\theoremseparator{}
\newtheorem{Theorem}{Theorem}[section]
\newtheorem{theorem}[Theorem]{Theorem}

\newtheorem{proposition}[Theorem]{Proposition}

\newtheorem{lemma}[Theorem]{Lemma}
\newtheorem{Corollary}[Theorem]{Corollary}
\newtheorem{corollary}[Theorem]{Corollary}

\theorembodyfont{\upshape}
\theoremsymbol{\ensuremath{\blacklozenge}}

\newtheorem{example}[Theorem]{Example}

\newtheorem{remark}[Theorem]{Remark}

\newtheorem{definition}[Theorem]{Definition}

\theoremstyle{nonumberplain}
\theoremheaderfont{\scshape}
\theorembodyfont{\normalfont}
\theoremsymbol{\ensuremath{\blacksquare}}

\newtheorem{proof}{Proof}
\qedsymbol{\ensuremath{_\blacksquare}}
\theoremclass{LaTeX}

\title{Noncommutative $\eps$-graded connections\footnote{Work
supported by ANR grant NT05-3-43374 ``GenoPhy''.}}
\author{Axel de Goursac$^{a,b}$, Thierry Masson$^a$, Jean-Christophe Wallet$^a$}
\date{}

\begin{document}

\maketitle
\vspace*{-1cm}
\begin{center}
\textit{$^a$Laboratoire de Physique Th\'eorique, B\^at.\ 210\\
    Universit\'e Paris XI,  F-91405 Orsay Cedex, France\\
    e-mail: \texttt{axelmg@melix.net}, \texttt{thierry.masson@u-psud.fr},
\texttt{jean-christophe.wallet@th.u-psud.fr}}\\[1ex]
\textit{$^b$Mathematisches Institut der Westf\"alischen
  Wilhelms-Universit\"at \\Einsteinstra\ss{}e 62, D-48149 M\"unster,
  Germany}\\
\end{center}%

\vskip 2cm

\begin{abstract}
We introduce the new notion of $\eps$-graded associative algebras which takes its root into the notion of commutation factors introduced in the context of Lie algebras \cite{Scheunert:1979}. We define and study the associated notion of $\eps$-derivation-based differential calculus, which generalizes the derivation-based differential calculus on associative algebras. A corresponding notion of noncommutative connection is also defined. We illustrate these considerations with various examples of $\eps$-graded algebras, in particular some graded matrix algebras and the Moyal algebra. This last example permits also to interpret mathematically a noncommutative gauge field theory.
\end{abstract}

\vfill
\begin{flushleft}
LPT-Orsay/08-96
\end{flushleft}

\pagebreak

\section{Introduction}

Noncommutative geometry finds its origins in the fact that spaces are dual to commutative algebras. For instance, a topological (Hausdorff, locally compact) space can be equivalently described by the algebra of its continuous functions, which is a commutative $C^\ast$-algebra. The noncommutative $C^\ast$-algebras would then correspond to some ``noncommutative topological spaces''. This correspondence also takes place in measure theory and in differential geometry, which gives rise to noncommutative geometry \cite{Connes:1994, Landi:1997}. This in particular has led in physics to theories describing the standard model of particle physics \cite{Chamseddine:2006ep} coupled to gravitation in the spectral triple approach, and others as candidates for new physics beyond the standard model.

While the spectral triple approach \cite{Connes:1994} is actually a natural way to construct noncommutative extension of differential calculus and is intensively studied presently, the differential calculus based on the derivations of an associative algebra provides another (not so widely explored) way to produce noncommutative differential calculus. The differential calculus based on the derivations of an associative noncommutative algebra has been introduced in \cite{DuboisViolette:1988cr} (see also \cite{DuboisViolette:1994cr, DuboisViolette:1999cj}). Besides, further developments with possible applications to the construction of noncommutative versions of Yang-Mills-Higgs type models have been carried out in \cite{DuboisViolette:1988ir, DuboisViolette:1989vq} for the algebra of matrix-valued functions and in  \cite{DuboisViolette:1998su, Masson:1999ea}.

On another side, graded algebra has been studied for a long time \cite{Nastasescu:1982}, and its most well-known applications are the theories of supermanifolds, and of graded Lie algebras \cite{Fuks:1986}. Note also that a recent work on graded associative algebras, and in particular graded matrix algebras, has been done \cite{Bahturin:2002}. A natural generalization of $\gZ$-grading (or $\gZ_2$-grading) is the $\Gamma$-grading, where $\Gamma$ is an abelian group. However, in order to recover similar properties as for $\gZ$-graded Lie algebras, one has to define an additional structure, introduced for instance in \cite{Bourbaki:2007} for associative algebras and studied by Scheunert \cite{Scheunert:1979} in the context of Lie algebras, called the commutation factor.

One can now ask for a noncommutative geometry adapted to this setting of graded algebras. It would then correspond to some ``graded noncommutative spaces''. Note that projective varieties are precisely such objects at the level of noncommutative algebraic geometry (see \cite{Artin:1994} and references therein). In this paper, we study the notion of $\eps$-graded associative algebra in relation to the derivation-based non commutative geometry. An $\eps$-graded algebra, as it is defined here, is an associative algebra endowed with a commutation factor $\eps$, to which one can associate canonically a structure of generalized Lie algebra \cite{Rittenberg:1978mr} or $\eps$-Lie algebra \cite{Scheunert:1979}. The main contribution of this paper is to generalize the defining features of the differential calculus based on derivations of an associative algebra to the case of $\eps$-graded algebras and their $\eps$-derivations. We show in particular that the differential calculus based on $\eps$-derivations of an $\eps$-graded associative algebra is a general and natural framework including all the works mentioned above. Finally, we illustrate these constructions on various examples.

The notion of $\eps$-graded algebras as it is defined and considered here is related to some constructions introduced in \cite[Chapter~III]{Bourbaki:2007} in order to study in an unifying way some particular graded algebras: the tensor, the symmetric and the exterior algebras associated to modules. Let us mention also that our definition is closed to some definitions of color (Lie) algebras, where a commutation factor is also used under the name of bicharacter. Our terminology is deliberately chosen from the one by Scheunert \cite{Scheunert:1979} who studied $\eps$-Lie algebras. This choice is motivated by the fact that some $\eps$-Lie algebras will be naturally associated to $\eps$-graded associative algebras, so that it is very convenient to use his nomenclature. 

\medskip
The paper is organized as follows. In subsection \ref{subsec-commfact}, we recall the notion of commutation factor and give relation to the theory of Schur's multipliers. We introduce the notion of an $\eps$-graded algebra and its $\eps$-derivations in subsection \ref{subsec-epsalg}, and then we construct the differential calculus based on $\eps$-derivations in subsection \ref{subsec-diffcalc} and its theory of (noncommutative) $\eps$-connections in subsection \ref{subsec-epsconn}. As an illustration, we apply this formalism in subsection \ref{subsec-excommut} to $\eps$-graded commutative algebras, for a supermanifold. Then, we study in subsections \ref{subsec-exelem} and \ref{subsec-exfine} the matrix algebra endowed with elementary and fine gradings, which produces interesting examples of $\eps$-graded algebras. We study the general properties of these $\eps$-graded matrix algebras, their differential calculus and their space of $\eps$-connections. Subsection \ref{subsec-moyal} is related to the case of the Moyal algebra. An $\eps$-graded algebra is constructed from the Moyal space. We apply this construction to find naturally the recently constructed candidate for a renormalizable gauge action on Moyal space \cite{deGoursac:2007gq} as built from a graded curvature.

\section{\texorpdfstring{Noncommutative geometry based on $\eps$-derivations}{Noncommutative geometry based on epsilon-derivations}}
\label{sec-ncg}

\subsection{Commutation factors}
\label{subsec-commfact}

We recall here the principal features of the commutation factors and the Schur's multipliers. For commutation factors, we refer mainly to the extended study by Scheunert \cite{Scheunert:1979} for the general features, and to \cite{Bourbaki:2007} for applications to associative algebras. Let $\gK$ be a field, $\gK^\ast$ its multiplicative group, and $\Gamma$ an abelian group.

\begin{definition}
\label{def-commfactor}
A {\it commutation factor} is a map $\eps:\Gamma\times\Gamma\to\gK^\ast$ satisfying: $\forall i,j,k\in\Gamma$,
\begin{subequations}
\label{eq-defepsilon}
\begin{align}
&\eps(i,j)\eps(j,i)=1_\gK\\
&\eps(i,j+k)=\eps(i,j)\eps(i,k)\\
&\eps(i+j,k)=\eps(i,k)\eps(j,k)
\end{align}
\end{subequations}
\end{definition}

Note that as trivial consequences, one has the following useful relations, $\forall i, j\in\Gamma$:
\begin{gather*}
\eps(i,0) = \eps(0, i) = 1_{\gK}\\
\eps(i,i) \in \{1_{\gK}, -1_{\gK}\}\\
\eps(j,i) = \eps(i,-j) = \eps(i,j)^{-1}.
\end{gather*}

A very simple and well-known example of such structure concerns the case $\Gamma = \gZ$ for which $\eps(p,q) = (-1)^{pq}$ is a commutation factor. In fact, from a general result which will be mentioned later on, this is the only non trivial commutation factor on $\gZ$, the trivial one being $\eps(p,q) = 1$. We will call this non trivial commutation factor the natural one on $\gZ$.

One can define an equivalence relation on the commutation factors of an abelian group in the following way:
\begin{definition}
Two commutation factors $\eps$ and $\eps'$ are called equivalent if there exists $f\in\Aut(\Gamma)$, the group of automorphisms of $\Gamma$, such that: $\forall i,j\in\Gamma$,
\begin{equation*}
\eps'(i,j)=f^\ast\eps(i,j)=\eps(f(i),f(j)).
\end{equation*}
\end{definition}

Thanks to the axioms \eqref{eq-defepsilon}, one defines the signature function $\psi_\eps:\Gamma\to\{-1_\gK,1_\gK\}$ of $\eps$ by $\forall i\in\Gamma$, $\psi_\eps(i)=\eps(i,i)$, which satisfies: $\forall f\in\Aut(\Gamma)$, $\psi_{f^\ast\eps}=f^\ast\psi_\eps$.
\begin{proposition}
Let $\eps$ be a commutation factor. Let us define
\begin{equation*}
\Gamma_\eps^0=\{i\in\Gamma,\quad\eps(i,i)=1_\gK\},\qquad \Gamma_\eps^1=\{i\in\Gamma,\quad\eps(i,i)=-1_\gK\},
\end{equation*}
then $\eps$ is called proper if $\Gamma_\eps^0=\Gamma$. This property is compatible with the equivalence relation on the commutation factors. If $\eps$ is not proper, $\Gamma_\eps^0$ is a subgroup of $\Gamma$ of index 2, and $\Gamma_\eps^0$ and $\Gamma_\eps^1$ are its residues.
\end{proposition}

In these notations, we define the signature factor of the commutation factor $\eps$ by: $\forall i,j\in\Gamma$,
\begin{itemize}
\item $s(\eps)(i,j)=-1_\gK$ if $i\in\Gamma_\eps^1$ and $j\in\Gamma_\eps^1$.
\item $s(\eps)(i,j)=1_\gK$ if not.
\end{itemize}
$s(\eps)$ is also a commutation factor, such that $\forall f\in\Aut(\Gamma)$, $s(f^\ast\eps)=f^\ast s(\eps)$.

\begin{lemma}
\label{lem-thierry}
Let $\eps_1$ and $\eps_2$ be two commutation factors respectively on the abelian groups $\Gamma_1$ and $\Gamma_2$, over the same field $\gK$. Then, the map $\eps$ defined by: $\forall i_1,j_1\in\Gamma_1$, $\forall i_2,j_2\in\Gamma_2$,
\begin{equation}
\eps((i_1,i_2),(j_1,j_2))=\eps_1(i_1,j_1)\eps_2(i_2,j_2),\label{eq-eps-product}
\end{equation}
is a commutation factor on the abelian group $\Gamma=\Gamma_1\times\Gamma_2$, over $\gK$.
\end{lemma}

\begin{proposition}
\label{prop-factcomm}
Let $\Gamma$ be a finitely generated abelian group and $\gK$ a field. Then, $\Gamma$ is the direct product of a finite number of cyclic groups, whose generators are denoted by $\{e_r\}_{r\in I}$. Any commutation factor on $\Gamma$ over $\gK$ takes the form: $\forall i=\sum_{r\in I}\lambda_re_r,j=\sum_{s\in I}\mu_se_s\in\Gamma$ ($\lambda_r,\mu_s\in\gZ$),
\begin{equation*}
\eps(i,j)=\prod_{r,s\in I}\eps(e_r,e_s)^{\lambda_r\mu_s},
\end{equation*}
with the condition: if $m_{rs}$ is the greatest common divisor of the orders $m_r\geq0$ of $e_r$ and $m_s\geq0$ of $e_s$ in $\Gamma$,
\begin{align*}
&\forall r\in I,\text{ such that } m_r\text{ is odd, }\eps(e_r,e_r)=1_\gK,\\
&\forall r\in I,\text{ such that } m_r\text{ is even, }\eps(e_r,e_r)=\pm 1_\gK,\\
&\forall r,s\in I,\quad \eps(e_r,e_s)^{m_{rs}}=1_\gK
\end{align*}
\end{proposition}
The proof is straightforward and has been given in \cite{Scheunert:1979}. This Proposition gives the explicit form of a commutation factor on a finitely generated abelian group, but this is not a classification of such factors. In general, it is not easy to obtain this classification and it is related to the theory of multipliers as we will see below, but in the following example, coming from \cite{Scheunert:1979}, things are simplifying.

\begin{example}
\label{ex-factcomm}
Let $\Gamma=\gZ_p^n$, with $p$ a prime number, $\gK$ a field of characteristic different from $p$, and $\alpha\neq 1_\gK$ a $p$th root of unity in $\gK$. Then, any commutation factor on $\Gamma$ over $\gK$ takes the unique following form: $\forall i,j\in\Gamma$,
\begin{equation*}
\eps(i,j)=\alpha^{\varphi(i,j)},
\end{equation*}
where $\varphi$ is a bilinear form on the vector space $\gZ_p^n$ over the field $\gZ_p$, which is:
\begin{itemize}
\item symmetric if $p=2$,
\item skew-symmetric if $p\geq3$. In this case, $\eps$ is proper.
\end{itemize}
To equivalent commutation factors correspond equivalent (in the sense of linear algebra) bilinear forms.
\end{example}
\begin{proof}
We recall here the proof of \cite{Scheunert:1979}. Let $\eps$ be a commutation factor on $\Gamma$ over $\gK$. Since $p$ is prime and different from the characteristic of $\gK$, there exists a $p$th root of unity in $\gK$, $\alpha\neq1_\gK$, and all the $p$th roots of unity are powers of $\alpha$. Using Proposition \ref{prop-factcomm}, if $\{e_r\}_{r\in I}$ are the canonical generators of $\Gamma=(\gZ_p)^n$, $\forall r,s\in I$, $\eps(e_r,e_s)^p=1_\gK$, so that there exists $m_{rs}\in\gZ_p$ such that $\eps(e_r,e_s)=\alpha^{m_{rs}}$ and $m_{sr}=-m_{rs}$.

Then, $\forall i=\sum_{r\in I}\lambda_re_r,j=\sum_{s\in I}\mu_se_s\in\Gamma$ ($\lambda_r,\mu_s\in\gZ$), $\eps(i,j)=\alpha^{\varphi(i,j)}$, where $\varphi(i,j)=\sum_{r,s\in I}\lambda_r\mu_sm_{rs}$ is a bilinear form. If $p\geq3$, $\forall r\in I$, $m_{rr}=0$ and $\varphi$ is skew-symmetric. If $p=2$, $\forall r,s\in I$, $m_{sr}=-m_{rs}=m_{rs}$ and $\varphi$ is symmetric.
\end{proof}

\subsubsection{Schur multipliers}
\label{subsub-schur}

We will now present the theory of multipliers of an abelian group, due to Schur (also for non-abelian groups), which is related to the classification of commutation factors. Let us now recall the standard definition of a factor set, closely related to a projective representation.
\begin{definition}
\label{def-multiplier}
Let $\Gamma$ be an abelian group, and $\gK$ a field. A {\it factor set} is an application $\sigma:\Gamma\times\Gamma\to\gK^\ast$ such that: $\forall i,j,k\in\Gamma$,
\begin{equation}
\sigma(i,j+k)\sigma(j,k)=\sigma(i,j)\sigma(i+j,k).\label{eq-eps-deffactorset}
\end{equation}
Two factor sets $\sigma$ and $\sigma'$ are said to be equivalent if there exists an application $\rho:\Gamma\to\gK^\ast$ such that: $\forall i,j\in\Gamma$,
\begin{equation*}
\sigma'(i,j)=\sigma(i,j)\rho(i+j)\rho(i)^{-1}\rho(j)^{-1}.
\end{equation*}
The quotient of the set of factor sets by this equivalence relation is an abelian group, for the product of $\gK$, and is called the multiplier group $M_\Gamma$ of $\Gamma$. Each class $[\sigma]\in M_\Gamma$ is called a {\it multiplier}.

If $\sigma$ is a factor set and $f\in\Aut(\Gamma)$, the pullback $f^\ast\sigma$ defined by: $\forall i,j\in\Gamma$,
\begin{equation*}
f^\ast\sigma(i,j)=\sigma(f(i),f(j)),
\end{equation*}
is also a factor set. Moreover, this operation is compatible with the above equivalence relation, so that the pullback can now be defined on the multipliers: $f^\ast[\sigma]=[f^\ast\sigma]$. This defines an equivalence relation on the multipliers, which is not compatible with the product.

A more refined equivalence relation involves also subgroups of $\Gamma$: if $[\sigma]$ and $[\sigma']$ are multipliers and $\Gamma_0$ and $\Gamma_1$ are subgroups of $\Gamma$, $([\sigma],\Gamma_0)$ and $([\sigma'],\Gamma_1)$ are called equivalent if there exists $f\in\Aut(\Gamma)$ such that $f(\Gamma_1)=\Gamma_0$ and $[\sigma']=f^\ast[\sigma]$.
\end{definition}

Note that the equation \eqref{eq-eps-deffactorset} is related to the definition of cocycles of the group $\Gamma$, while the equivalence of factor sets can be reexpressed in terms of coboundaries, so that multipliers are in fact exponentials of the cohomology classes $H^{2}(\Gamma,\gK)$ of the group $\Gamma$.

\subsubsection{Relation between commutation factors and multipliers}
\label{subsub-schurcommfact}

If $\eps$ is a commutation factor on $\Gamma$ over $\gK$, notice that, due to Definition \ref{def-commfactor}, it is a factor set of $\Gamma$. But there is a deeper relation between commutation factors and factor sets, given by the following theorem \cite{Scheunert:1979}:
\begin{theorem}
\label{thm-multcomm}
Let $\Gamma$ be an abelian group, and $\gK$ a field.
\begin{itemize}
\item Any multiplier $[\sigma]$ defines a unique proper commutation factor $\eps_\sigma$ on $\Gamma$ by: $\forall i,j\in\Gamma$,
\begin{equation}
\eps_\sigma(i,j)=\sigma(i,j)\sigma(j,i)^{-1}.\label{eq-multcomm}
\end{equation}
\item If $\Gamma$ is finitely generated, any proper commutation factor $\eps$ on $\Gamma$ can be constructed from a multiplier $[\sigma]$ by \eqref{eq-multcomm}.
\item If, in addition, $\gK$ is algebraically closed, then $\eps$ is constructed from a unique multiplier $[\sigma]$.
\item For $\Gamma$ finitely generated and $\gK$ algebraically closed, if two proper commutation factors $\eps_\sigma$ and $\eps_{\sigma'}$ are equivalent, then $[\sigma']$ is a pullback of $[\sigma]$.
\end{itemize}
\end{theorem}
\begin{proof}
This theorem has been partly proved in \cite{Scheunert:1979}. We recall and complete here the proof.
\begin{itemize}
\item Let $\sigma$ be a factor set, and $\eps_\sigma(i,j)=\sigma(i,j)\sigma(j,i)^{-1}$, for $i,j\in\Gamma$.\\
Then, $\eps_\sigma(i,j)\eps_\sigma(j,i)=1_\gK$, and $\forall i,j,k\in\Gamma$,
\begin{multline*}
\eps_\sigma(i,j+k)\eps_\sigma(i,j)^{-1}\eps_\sigma(i,k)^{-1}=\\
\sigma(i,j+k)\sigma(i,j)^{-1}\sigma(i,k)^{-1} (\sigma(j+k,i)\sigma(j,i)^{-1}\sigma(k,i)^{-1})^{-1}.
\end{multline*}
By using three times the property \eqref{eq-eps-deffactorset}, one obtains
\begin{equation*}
\sigma(j+k,i)=\sigma(i,j+k)\sigma(j,i)\sigma(k,i)\sigma(i,j)^{-1}\sigma(i,k)^{-1},
\end{equation*}
which proves that $\eps_\sigma(i,j+k)=\eps_\sigma(i,j)\eps_\sigma(i,k)$. The third axiom of Definition \ref{def-commfactor} can be proved in a similar way. If $\sigma'(i,j)=\sigma(i,j)\rho(i+j)\rho(i)^{-1}\rho(j)^{-1}$, then $\sigma'(i,j)\sigma'(j,i)^{-1}=\sigma(i,j)\sigma(j,i)^{-1}$.
\item Let us suppose that $\Gamma$ is finitely generated, and $\{e_r\}_{r\in I}$ a system of generators, where $I$ is an ordered (finite) set. Let $\eps$ be a proper commutation factor on $\Gamma$. Define $\sigma:\Gamma\times\Gamma\to\gK^\ast$ by $\forall i=\sum_{r\in I}\lambda_re_r,j=\sum_{s\in I}\mu_se_s\in\Gamma$ ($\lambda_r,\mu_s\in\gZ$),
\begin{equation*}
\sigma(i,j)=\prod_{r<s}\eps(e_r,e_s)^{\lambda_r\mu_s}.
\end{equation*}
Since $\forall r,s\in I$, $\eps(e_r,e_r)=1_\gK$ and $\eps(e_r,e_s)=\eps(e_s,e_r)^{-1}$, one has\\
$\eps(i,j)=\sigma(i,j)\sigma(j,i)^{-1}$. And $\forall k=\sum_{r\in I}\nu_re_r\in\Gamma$ ($\nu_r\in\gZ$),
\begin{equation*}
\sigma(i,j+k)\sigma(j,k)=\prod_{r<s}\eps(e_r,e_s)^{\lambda_r\mu_s+\lambda_r\nu_s+\mu_r\nu_s}=\sigma(i,j)\sigma(i+j,k).
\end{equation*}
\item Let $\Gamma$ be a finitely generated abelian group, $\gK$ an algebraically closed field, and $\eps$ a commutation factor on $\Gamma$ over $\gK$. Suppose that $\eps$ is constructed through \eqref{eq-multcomm} from two factor sets $\sigma$ and $\sigma'$. Then $\forall i,j\in\Gamma$, $\sigma(i,j)\sigma'(i,j)^{-1}=\sigma(j,i)\sigma'(j,i)^{-1}$, which means that $\sigma\sigma'^{-1}$ is a symmetric factor set. Since $\gK$ is algebraically closed, $\sigma\sigma'^{-1}$ is equivalent to one, and $\sigma$ and $\sigma'$ are in the same multiplier.
\item If $\Gamma$ is finitely generated and $\gK$ algebraically closed, consider two equivalent commutation factors $\eps$ and $\eps'$: $\eps'=f^\ast\eps$, with $f\in\Aut(\Gamma)$. Then, there exists a unique multiplier $[\sigma]$ such that $\eps=\eps_\sigma$. $\forall i,j\in I$,
\begin{equation*}
\eps'(i,j)=\sigma(f(i),f(j))\sigma(f(j),f(i))^{-1}=\eps_{f^\ast\sigma}(i,j).
\end{equation*}
By unicity of the associated multiplier of $\eps'$, we obtain the result.
\end{itemize}
\end{proof}

\begin{Corollary}
Let $\Gamma$ be a finitely generated abelian group, and $\gK$ an algebraically closed field. Then,
\begin{itemize}
\item the proper commutation factors on $\Gamma$ are classified by the equivalence classes (by pullback) of the multipliers of $\Gamma$.
\item the non-proper commutation factors on $\Gamma$ are classified by the equivalence classes of multipliers and subgroups of index 2 of $\Gamma$.
\end{itemize}
\end{Corollary}
\begin{proof}
The first point is a direct consequence of Theorem~\ref{thm-multcomm}.

For the second point, if $\eps_1$ and $\eps_2$ are non-proper equivalent commutation factors on $\Gamma$, then there exists $f\in\Aut(\Gamma)$ such that $\eps_2=f^\ast\eps_1$, and $f(\Gamma_{\eps_2}^0)=\Gamma_{\eps_1}^0$. $\forall \alpha=1,2$, we decompose $\eps_\alpha=s(\eps_\alpha)\widetilde\eps_\alpha$, with $\widetilde\eps_\alpha$ proper commutation factors. Using now Theorem~\ref{thm-multcomm}, there exist unique multipliers $[\sigma_\alpha]$ such that $\widetilde\eps_\alpha=\eps_{\sigma_\alpha}$, and they satisfy $[\sigma_2]=f^\ast[\sigma_1]$. Then, $([\sigma_1],\Gamma_{\eps_1}^0)$ and $([\sigma_2],\Gamma_{\eps_2}^0)$ are equivalent.

Conversely, if $\forall\alpha=1,2$, $\eps_\alpha=s(\eps_\alpha)\eps_{\sigma_\alpha}$, with $[\sigma_\alpha]$ multipliers such that $([\sigma_1],\Gamma_{\eps_1}^0)$ and $([\sigma_2],\Gamma_{\eps_2}^0)$ are equivalent, then there exists $f\in\Aut(\Gamma)$ such that $f(\Gamma_{\eps_2}^0)=\Gamma_{\eps_1^0}$ and $[\sigma_2]=f^\ast[\sigma_1]$. It means that $s(\eps_2)=f^\ast s(\eps_1)$ and $\eps_{\sigma_2}=f^\ast\eps_{\sigma_1}$, so that $\eps_2=f^\ast\eps_1$.
\end{proof}

Thus, by studying the equivalence classes of multipliers and of multipliers and subgroups of index 2 of a given abelian group $\Gamma$, one can characterize all the commutation factors on $\Gamma$ up to equivalence.

\subsection{\texorpdfstring{$\eps$-graded associative algebras}{epsilon-graded associative algebras}}
\label{subsec-epsalg}

Let $\gK$ be a field, $\gK^\ast$ its multiplicative group, and $\Gamma$ an abelian group. The latter notion of commutation factors has been introduced in the context of graded Lie algebras by Scheunert and gives rise to the following definition:
\begin{definition}[$\eps$-Lie algebra, \cite{Scheunert:1979}]
Let $\kg^\bullet$ be a $\Gamma$-graded vector space, $\eps$ a commutation factor on $\Gamma$, and $[-,-]_\eps:\kg^\bullet\times\kg^\bullet\to\kg^\bullet$ a bilinear product homogeneous of degree 0 satisfying
\begin{align}
[a,b]_\eps &=-\eps(|a|,|b|)[b,a]_\eps\\
[a,[b,c]_\eps]_\eps &=[[a,b]_\eps,c]_\eps+\eps(|a|,|b|)[b,[a,c]_\eps]_\eps,\label{epsjacobi}
\end{align}
$\forall a,b,c\in\kg^\bullet$ homogeneous, where $|a|\in\Gamma$ is the degree of $a\in\kg^\bullet$. The couple $(\kg^\bullet, [-,-]_\eps)$ is called an $\eps$-Lie algebra.

An $\eps$-Lie algebra for which the product $[-,-]_\eps$ is always $0$ will be called an abelian $\eps$-Lie algebra.
\end{definition}

Notice that the notion of $\eps$-Lie algebras has been generalized \cite{Larsson:2005} in quasi-hom-Lie algebras.

\bigskip
In this paper, following \cite{Bourbaki:2007}, we introduce the following structure (the terminology is chosen to be close to the one by Scheunert):
\begin{definition}[$\eps$-graded associative algebra]
Let $\algrA$ be an associative unital $\Gamma$-graded $\gK$-algebra, endowed with a commutation factor $\eps$ on $\Gamma$, then $(\algrA,\eps)$ will be called an $\eps$-graded (associative) algebra.
\end{definition}

Notice that the $\eps$-structure is only related to the algebra $\algrA$ through the grading abelian group $\Gamma$. In particular, the product in the algebra is not connected to this structure. In the following, such an $\eps$-graded (associative) algebra will be denoted simply by $\algrA_\eps$ or even $\algrA$ if no confusion arises.

Any $\gZ$-graded associative algebra is an $\eps$-graded associative algebra for the natural commutation factor on $\gZ$, so that the theory described below can be applied to any $\gZ$-graded (associative) algebra.

\begin{remark}
\label{remark-bigrading}
Using Lemma~\ref{lem-thierry}, if $\algA^{\bullet, \bullet}$ is an associative unital $\Gamma_1\times\Gamma_2$-bigraded $\gK$-algebra equipped with two commutation factors $\eps_1$ and $\eps_2$ for the two gradings separately, then it is also a $\eps$-graded algebra for the product grading $\Gamma_1\times\Gamma_2$ with $\eps$ defined by \eqref{eq-eps-product}.
\end{remark}

If $\algrA$ is an $\eps$-graded algebra, one can construct its underlying $\eps$-Lie algebra using the bracket defined by
\begin{equation}
[a,b]_\eps=a\fois b-\eps(|a|,|b|)\ b\fois a.
\end{equation}
$\forall a,b\in\algrA$ homogeneous. We will denote by $\algrA_{\Lie, \eps}$ this structure.

\begin{definition}[$\eps$-graded commutative algebra]
$\algrA_\eps$ is called an $\eps$-graded commutative algebra if $\algrA_{\Lie, \eps}$ is an abelian $\eps$-Lie algebra.
\end{definition}

For the case of $\gZ$-graded algebras, depending on the commutation factor, one obtains as $\eps$-graded commutative algebras either commutative and graded algebras (for the trivial commutation factor) or graded commutative algebras (for the natural commutation factor).

\begin{definition}[$\eps$-center]
Let $\algrA_\eps$ be an $\eps$-graded algebra. The $\eps$-center of $\algrA_\eps$ is the $\eps$-graded commutative algebra
\begin{equation}
\caZ^\bullet_\eps(\algA)=\{a\in\algrA,\ \forall b\in\algrA\ [a,b]_\eps=0\}.
\end{equation}
\end{definition}

Depending on the choice of the $\eps$-structure on a $\Gamma$-graded algebra, this $\eps$-center can be very different. Nevertheless, one has: $|\gone|=0$ and $\gone\in\caZ^\bullet_\eps(\algA)$.

Let us now mention some elementary constructions using $\eps$-graded algebras. Let $\algrA$ and $\algrB$ be two $\eps$-graded algebras with the same commutation factor $\eps$.

A morphism of $\eps$-graded algebras is defined to be a morphism of associative unital $\Gamma$-graded algebras $\chi:\algrA\to\algrB$. As a consequence, $\chi$ is also a morphism of $\eps$-Lie algebras between $\algrA_{\Lie, \eps}$ and $\algrB_{\Lie, \eps}$.

As in \cite{Bourbaki:2007}, the $\eps$-graded tensor product of two $\eps$-graded algebras $\algrA$ and $\algrB$ is the $\eps$-graded algebra defined  as the $\Gamma$-graded vector space $(\algA \otimes \algB)^\bullet$ for the total grading, equipped with the product given by
\begin{equation}
(a\otimes b)\fois(c\otimes d)=\eps(|b|,|c|)(a\fois c)\otimes(b\fois d).
\end{equation}
$\forall a,c\in\algrA$ and $\forall b,d\in\algrB$ homogeneous.

\begin{lemma}
\label{lem-eps-gradedtensorproduct}
Let $\algrA_\eps$ and $\algrB_\eps$ be two $\eps$-graded commutative algebra. Then their $\eps$-graded tensor product is a $\eps$-graded commutative algebra.
\end{lemma}
When no confusion will arise, we will make reference to this $\eps$-graded tensor product as the tensor product of $\eps$-graded algebras.

An $\eps$-trace on $\algrA$ is a linear map $T:\algA\to\gK$, which satisfies
\begin{equation}
T(a\fois b)=\eps(|a|,|b|)T(b\fois a).\label{deftrace}
\end{equation}
$\forall a,b\in\algrA$ homogeneous.

\bigskip
The structure of module compatible with an $\eps$-graded algebra $\algrA$ is simply the structure of $\Gamma$-graded module. $\modrM$ is a $\Gamma$-graded module on $\algrA$ if it is a $\Gamma$-graded vector space and a module on $\algrA$ such that $\modM^i\algA^j\subset \modM^{i+j}$ (for right modules) $\forall i,j\in\Gamma$. The space of homomorphisms of $\modrM$ is an $\eps$-graded algebra and will be denoted by $\Hom^\bullet_\algA(\modM,\modM)$.

\bigskip
Let us now introduce the key object which will permit us to introduce a differential calculus adapted to this situation.

\begin{definition}[$\eps$-derivations]
An $\eps$-derivation on the $\eps$-graded algebra $\algrA$ is a linear map $\kX: \algrA \to \algrA$ such that, if $\kX$ is homogeneous (as map of graded spaces) of degree $|\kX|\in\Gamma$,
\begin{equation}
\kX(a\fois b)=\kX(a)\fois b+\eps(|\kX|,|a|)\ a\fois\kX(b),\label{defderiv}
\end{equation}
$\forall a,b\in\algA$ with $a$ homogeneous.

We denote by $\Der^\bullet_\eps(\algA)$ the $\Gamma$-graded space of $\eps$-derivation on the $\eps$-graded algebra $\algrA$.
\end{definition}

Notice that this definition makes explicit reference to the $\eps$-structure, so that $\Der^\bullet_\eps(\algA)$ really depends on it. Moreover, the $\eps$-derivations are a particular case of the notion of $(\sigma,\tau)$-derivations \cite{Hartwig:2006}.

\begin{proposition}[Structure of $\Der^\bullet_\eps(\algA)$]
The space $\Der^\bullet_\eps(\algA)$ is an $\eps$-Lie algebra for the bracket
\begin{equation*}
[\kX,\kY]_\eps = \kX\kY - \eps(|\kX|,|\kY|) \kY \kX.
\end{equation*}

It is also a $\caZ^\bullet_\eps(\algA)$-bimodule for the product structure
\begin{equation}
(z\fois \kX)(a)=\eps(|z|,|\kX|)(\kX\fois z)(a)=z\fois (\kX(a)),\label{defmodder}
\end{equation}
$\forall\kX\in\Der^\bullet_\eps(\algA)$, $\forall z\in\caZ^\bullet_\eps(\algA)$ and $\forall a\in\algrA$ homogeneous.
\end{proposition}

Notice that the left and right module structures are equivalent modulo the factor $\eps(|z|,|\kX|)$. So that we will mention it as a module structure, not as a bimodule one. In order to take into account this extra factor, it would be convenient to introduce the notion of $\eps$-central bimodule, as a straightforward adaptation of the notion of central bimodule defined in \cite{DuboisViolette:1994cr} and \cite{DuboisViolette:1995tz}. We will not go further in this direction here.

As usual, an inner $\eps$-derivation on $\algrA$ is an $\eps$-derivation which can be written as
\begin{equation*}
b \mapsto \ad_a(b) = [a,b]_\eps,
\end{equation*}
for an $a \in \algrA$. We denote by
\begin{equation*}
\Int^\bullet_\eps(\algA)=\{\ad_a,\ a\in\algrA\}
\end{equation*}
the space of inner $\eps$-derivations on $\algrA$.

\begin{proposition}
$\Int^\bullet_\eps(\algA)$ is an $\eps$-Lie ideal and a $\caZ^\bullet_\eps(\algA)$-module.

This permits one to define the quotient
\begin{equation*}
\Out^\bullet_\eps(\algA)=\Der^\bullet_\eps(\algA)/\Int^\bullet_\eps(\algA)
\end{equation*}
as an $\eps$-Lie algebra and a $\caZ^\bullet_\eps(\algA)$-module. This is the space of outer $\eps$-derivations on $\algrA$.
\end{proposition}

From these considerations, one then gets the two short exact sequences of $\eps$-Lie algebras and $\caZ_\eps^\bullet(\algA)$-modules:
\begin{gather*}
\xymatrix@1@C=25pt{{\algzero} \ar[r] & {\caZ^\bullet_\eps(\algA)} \ar[r] & {\algrA} \ar[r]^-{\ad} & {\Int^\bullet_\eps(\algA)} \ar[r] & {\algzero}}\\[3mm]
\xymatrix@1@C=25pt{{\algzero} \ar[r] & {\Int^\bullet_\eps(\algA)} \ar[r] & {\Der^\bullet_\eps(\algA)} \ar[r] & {\Out^\bullet_\eps(\algA)} \ar[r] & {\algzero}}
\end{gather*}

It is possible to define the notion of $\eps$-derivations on $\algrA$ taking values in a bimodule. We will not make use of it in the following.

\bigskip
In the rest of this subsection, all the linear structures will take place over the field $\gK=\gC$.

Let us now take a look at the notion of involution in this framework. In order to do that, one has to further constrain the notion of commutation factors, which we require to satisfy the hermitean condition
\begin{equation}
\overline{\eps(i,j)}=\eps(j,i)\label{eq-herm-factcom}
\end{equation}
$\forall i,j\in\Gamma$. Equivalently, $|\eps(i,j)| = 1$, $\forall i,j\in\Gamma$

\begin{definition}[Involution, unitarity and reality]
An involution on an $\eps$-graded algebra $\algrA$ (over the field $\gK=\gC$) is an antilinear map $\algA^i\to\algA^{-i}$, for any $i\in\Gamma$, denoted by $a\mapsto a^\ast$, such that
\begin{equation}
\label{eq-involution}
(a^\ast)^\ast=a,\qquad (a\fois b)^\ast=b^\ast\fois a^\ast.
\end{equation}
$\forall a,b\in\algrA$.

The unitary group of $\algrA$ associated to $\ast$ is defined and denoted by
\begin{equation*}
\caU(\algA)=\{g\in\algrA,\ g^\ast\fois g=\gone\}.
\end{equation*}

An $\eps$-derivation decomposed in its homogeneous elements as $\kX=\sum_{k\in\Gamma}\kX_k$, where $|\kX_k|=k$, is called real if, $\forall a\in\algrA$ homogeneous, $\forall k\in\Gamma$, $(\kX_k(a))^\ast=\eps(|a|,|\kX_k|)\kX_{-k}(a^\ast)$.

An $\eps$-trace $T$ on $\algrA$ is real if, $\forall a\in\algrA$, $T(a^\ast)=\overline{T(a)}$.
\end{definition}

Finally, a hermitean structure on a right $\algrA$-module $\modrM$ is a sesquilinear form $\langle-,-\rangle:\modM^i\times\modM^j\to\algA^{j-i}$, for $i,j\in\Gamma$, such that $\forall m,n\in\modrM$, $\forall a,b\in\algrA$,
\begin{equation}
\label{eq-hermiteanstructure}
\langle m,n\rangle^\ast=\langle n,m\rangle,\qquad \langle ma,nb\rangle=a^\ast\langle m,n\rangle b.
\end{equation}

\begin{remark}[Involution and $\eps$-structure]
Notice that the definition of the involution as it is given by \eqref{eq-involution} does not make direct reference to the $\eps$-structure, as well as the hermitean structure defined by \eqref{eq-hermiteanstructure}. In Section~\ref{sec-ex}, some examples of $\eps$-graded matrix algebras are exposed, for which the usual involution is compatible with the present definition.

An other natural definition of the involution can be considered, where an explicit reference to $\eps$ appears in the second relation of \eqref{eq-involution} and in the two relations of \eqref{eq-hermiteanstructure}.
\end{remark}

\subsection{\texorpdfstring{The differential calculus based on $\eps$-derivations}{The differential calculus based on epsilon-derivations}}
\label{subsec-diffcalc}

The derivation-based differential calculus has been introduced in \cite{DuboisViolette:1988cr,DuboisViolette:1994cr,DuboisViolette:1995tz}. It has been studied for various algebras in \cite{DuboisViolette:1988ir, DuboisViolette:1989vq, Masson:1996, Masson:1995ph, DuboisViolette:1998su, Masson:1999ea, Wallet:2007em, Cagnache:2008tz, Wallet:2008bq} (see also \cite{DuboisViolette:1991,DuboisViolette:1999cj,Masson:2007jb,Masson:2008uq} for reviews) and some propositions to generalize this construction to graded algebras have been presented for instance in \cite{Lecomte:1992, Grosse:1999}.

\begin{remark}
For an associative algebra $\algA$, we recall that the derivation-based differential calculus is given by
\begin{equation}
\underline\Omega_{\Der}^n(\algA)=\{\omega:\Der(\algA)^n\to\algA,\ \caZ(\algA)\text{-multilinear and antisymmetric}\}.
\end{equation}
Thus, two important notions enter in this definition: the notions of derivations and of center (related to the bracket of the underlying Lie algebra of $\algA$). We have seen that the center and the derivations can be generalized into the center of quasi-hom-Lie algebras \cite{Larsson:2005} and the $(\sigma,\tau)$-derivations \cite{Hartwig:2006}, but in this paper, we will only consider the notions of $\eps$-center and $\eps$-derivations, since they are compatible together. Indeed, the same commutation factor is used for the bracket and the Leibniz relation, and therefore the framework of $\eps$-graded algebras seems to be well adapted for the definition of a derivation-based differential calculus.
\end{remark}
In the notations of the last subsection, let $\algrA$ be an $\eps$-graded algebra and $\modrM$ a right $\algrA$-module.

\begin{definition}
For $n\in\gN$ and $k\in\Gamma$, let $\Omr^{n,k}(\algA,\modM)$ be the space of $n$-linear maps $\omega$ from $(\Der^\bullet_\eps(\algA))^n$ to $\modrM$, such that $\forall\kX_1,\dots,\kX_n\in\Der^\bullet_\eps(\algA)$, $\forall z\in\caZ^\bullet_\eps(\algA)$ homogeneous,
\begin{align}
\omega(\kX_1,\dots,\kX_n) &\in\modM^{k+|\kX_1|+\dots+|\kX_n|},\nonumber\\
\omega(\kX_1,\dots,\kX_n z)&=\omega(\kX_1,\dots,\kX_n)z,\nonumber\\
\omega(\kX_1,\dots,\kX_i,\kX_{i+1},\dots,\kX_n)&=-\eps(|\kX_i|,|\kX_{i+1}|)\omega(\kX_1,\dots,\kX_{i+1},\kX_i,\dots,\kX_n),\label{defcalcdiff}
\end{align}
and $\Omr^{0,k}(\algA,\modM)=\modM^k$, where $\kX\fois z$ was defined by \eqref{defmodder}.
\end{definition}

From this definition, it follows that the vector space $\Omr^{\bullet,\bullet}(\algA,\modM)$ is $\gN\times\Gamma$-graded, and $\Omr^{n,\bullet}(\algA,\modM)$ is a right module on $\algrA$.

In the case $\modrM=\algrA$, we will use the notation $\Omr^{\bullet,\bullet}(\algA) = \Omr^{\bullet,\bullet}(\algA,\algA)$.

\begin{proposition}
\label{propproddiff}
Endowed with the following product and differential, $\Omr^{\bullet,\bullet}(\algA)$ is a differential algebra: $\forall\omega\in\Omr^{p,|\omega|}(\algA)$, $\forall\eta\in\Omr^{q,|\eta|}(\algA)$, and $\forall\kX_1,\dots,\kX_{p+q}\in\Der^\bullet_\eps(\algA)$ homogeneous, the product is
\begin{equation}
(\omega\fois\eta)(\kX_1,\dots,\kX_{p+q}) =\frac{1}{p!q!}\sum_{\sigma\in\kS_{p+q}}(-1)^{|\sigma|}f_1\ \omega(\kX_{\sigma(1)},\dots,\kX_{\sigma(p)})\fois\eta(\kX_{\sigma(p+1)},\dots,\kX_{\sigma(p+q)}),\label{defprod}
\end{equation}
and the differential is
\begin{multline}
\dd\omega(\kX_1,\dots,\kX_{p+1})=\sum_{m=1}^{p+1} (-1)^{m+1}f_2\ \kX_m\omega(\kX_1,\dots \omi{m} \dots,\kX_{p+1})\\
+\sum_{1\leq m<n\leq p+1}(-1)^{m+n}f_3\ \omega([\kX_m,\kX_n]_\eps,\dots \omi{m} \dots \omi{n} \dots,\kX_{p+1}),\label{defdiff}
\end{multline}
where the factors $f_i$ are given by
\begin{align}
f_1&=\prod_{m<n,\sigma(m)>\sigma(n)}\eps(|\kX_{\sigma(n)}|,|\kX_{\sigma(m)}|)\prod_{m\leq p} \eps(|\eta|,|\kX_{\sigma(m)}|)\nonumber\\
f_2&=\eps(|\omega|,|\kX_m|)\prod_{a=1}^{m-1}\eps(|\kX_a|,|\kX_m|)\nonumber\\
f_3&=\eps(|\kX_n|,|\kX_m|)\prod_{a=1}^{m-1}\eps(|\kX_a|,|\kX_m|)\prod_{a=1}^{n-1}\eps(|\kX_a|,|\kX_n|).\nonumber
\end{align}

$\Omr^{\bullet,\bullet}(\algA)$ is an $\widetilde{\eps}$-graded algebra for the abelian group $\widetilde{\Gamma}=\gZ\times\Gamma$ and the commutation factor $\widetilde{\eps}((p,i),(q,j))=(-1)^{pq}\eps(i,j)$. Furthermore, $\dd$ is an $\widetilde{\eps}$-derivation of $\Omr^{\bullet,\bullet}(\algA)$ of degree $(1,0)$ satisfying $\dd^2=0$.

For any right $\algrA$-module $\modrM$, $\Omr^{\bullet,\bullet}(\algA,\modM)$ is a right module on $\Omr^{\bullet,\bullet}(\algA)$ for the following action:
$\forall\omega\in\Omr^{p,|\omega|}(\algA,\modM)$, $\forall\eta\in\Omr^{q,|\eta|}(\algA)$, and $\forall\kX_1,\dots,\kX_{p+q}\in\Der^\bullet_\eps(\algA)$ homogeneous,
\begin{equation}
(\omega\eta)(\kX_1,\dots,\kX_{p+q}) =\frac{1}{p!q!}\sum_{\sigma\in\kS_{p+q}}(-1)^{|\sigma|}f_1\ \omega(\kX_{\sigma(1)},\dots,\kX_{\sigma(p)})\eta(\kX_{\sigma(p+1)},\dots,\kX_{\sigma(p+q)}),\label{defmodomega}
\end{equation}
where $f_1$ is still the one given above.
\end{proposition}

\begin{proof}
It is a straightforward calculation to check that \eqref{defprod} and \eqref{defdiff} are compatible and turning $\Omr^{\bullet,\bullet}(\algA)$ into a $\widetilde{\Gamma}$-graded differential algebra, in the same way that \eqref{defmodomega} turns $\Omr^{\bullet,\bullet}(\algA,\modM)$ into a right $\Omr^{\bullet,\bullet}(\algA)$-module. Then, Lemma~\ref{lem-thierry} and Remark~\ref{remark-bigrading}  give the end of the proposition.
\end{proof}

\begin{remark}[Notion of $\widetilde{\Gamma}$-graded differential algebra]
In the previous Proposition, we made use of the terminology ``$\widetilde{\Gamma}$-graded differential algebra'', or more simply ``differential algebra'', without defining precisely these notions. In this case, the meaning is obviously to denote the differential $\dd$ in the $\gZ$-direction on the $\widetilde{\Gamma}$-graded algebra, which satisfies the following properties: it is an $\widetilde{\eps}$-derivation of degree $(1,0)$ such that $\dd^2=0$.

It could be convenient and desirable to define a good notion of $\eps$-graded differential algebra in general.  In order to do that, one starts with an $\eps$-graded associative algebra $\algrA$ and an associated differential $\dd$ which is an $\eps$-derivation whose square is zero. But the degree of this differential is difficult to constrain as ``minimal'' in some ``direction'' in the associated grading group $\Gamma$, except for special situations, as the ones where this group is freely generated by a finite number of generators. In these cases, a differential could be required to have as degree one of these generators. Then, there could be as many differentials as there are generators, a situation very similar to the one for usual bigraded algebras ($\Gamma = \gZ \times \gZ$).
\end{remark}

In low degrees (in the $\gZ$ part), the expression of the differential takes the following form, $\forall a\in\algrA$, $\forall\omega\in\Omr^{1,\bullet}(\algA)$, $\forall\kX,\kY\in\Der^\bullet_\eps(\algA)$ homogeneous,
\begin{align}
\dd a(\kX)&=\eps(|a|,|\kX|)\kX(a),\nonumber\\
\dd\omega(\kX,\kY)&=\eps(|\omega|,|\kX|)\kX(\omega(\kY))-\eps(|\omega|+|\kX|,|\kY|)\kY(\omega(\kX))-\omega([\kX,\kY]_\eps).
\end{align}

One has to be aware of the fact that even if the $\eps$-Lie algebra of derivations is finite dimensional, the vector space $\Omr^{\bullet,\bullet}(\algA)$ can be infinite dimensional. For instance, this is indeed the case for Example~\ref{ex-supermatrix}, where the symmetric part of this differential calculus amounts to an infinite dimensional part.

\begin{proposition}
If $\algrA$ is an $\eps$-graded commutative algebra then $\Omr^{\bullet,\bullet}(\algA)$ is an $\widetilde{\eps}$-graded commutative algebra.
\end{proposition}

\begin{proof}
This is a straightforward computation.
\end{proof}

\begin{definition}[Restricted differential calculus]
Let $\kg^\bullet$ be an $\eps$-Lie subalgebra of $\Der^\bullet_\eps(\algA)$ and a module on $\caZ^\bullet_\eps(\algA)$. The restricted differential calculus $\Omr^{\grast,\grast}(\algA|\kg)$ associated to $\kg^\bullet$ is defined as the space of $n$-linear maps $\omega$ from $(\kg^\bullet)^n$ to $\algrA$ satisfying the axioms \eqref{defcalcdiff}, with the above product \eqref{defprod} and differential \eqref{defdiff}. It is also a $\widetilde\Gamma$-graded differential algebra.
\end{definition}

\begin{proposition}[Cartan operation]
Let $\kg^\bullet$ be an $\eps$-Lie subalgebra of $\Der_\eps^\bullet(\algA)$. $\kg^\bullet$ defines canonically a Cartan operation on $(\Omr^{\bullet,\bullet}(\algA),\dd)$ in the following way.

For each $\kX\in\Der_\eps^\bullet(\algA)$, the inner product with $\kX$ is the map $i_\kX:\Omr^{n,k}(\algA)\to\Omr^{n-1,k+|\kX|}(\algA)$ such that, $\forall\omega\in\Omr^{n,|\omega|}(\algA)$ and $\forall\kX,\kX_1,\dots,\kX_{n-1}\in\kg^\bullet$ homogeneous,
\begin{equation}
i_\kX\omega(\kX_1,\dots,\kX_{n-1})=\eps(|\kX|,|\omega|)\omega(\kX,\kX_1,\dots,\kX_{n-1}),\label{definnerprod}
\end{equation}
and $i_\kX\Omr^{0,\grast}(\algA)=\algzero$. $i_\kX$ is then an $\widetilde{\eps}$-derivation of the algebra $\Omr^{\bullet,\bullet}(\algA)$ of degree $(-1,|\kX|)$.

The associated Lie derivative $L_\kX$ to $i_\kX$ is
\begin{equation}
L_\kX=[i_\kX,\dd]=i_\kX\dd+\dd i_\kX:\Omr^{n,k}(\algA)\to\Omr^{n,k+|\kX|}(\algA),\label{deflieder}
\end{equation}
which makes it into an $\widetilde{\eps}$-derivation of $\Omr^{\bullet,\bullet}(\algA)$ of degree $(0,|\kX|)$, where the bracket in \eqref{deflieder} comes from the commutation factor $\widetilde{\eps}$ of $\Omr^{\bullet,\bullet}(\algA)$.

Then, the following properties are satisfied, $\forall\kX,\kY\in\kg^\bullet$ homogeneous,
\begin{align}
[i_\kX,i_\kY]&=i_\kX i_\kY+\eps(|\kX|,|\kY|)i_\kY i_\kX=0, &[L_\kX,i_\kY]&=L_\kX i_\kY-\eps(|\kX|,|\kY|)i_\kY L_\kX=i_{[\kX,\kY]_\eps},\nonumber\\
[L_\kX,\dd]&=L_\kX\dd-\dd L_\kX= 0, &[L_\kX,L_\kY]&=L_\kX L_\kY-\eps(|\kX|,|\kY|)L_\kY L_\kX=L_{[\kX,\kY]_\eps}.
\end{align}
\end{proposition}

\begin{proof}
$\forall\omega\in\Omr^{n,|\omega|}(\algA)$, $\forall\eta\in\Omr^{p,|\eta|}(\algA)$ and $\forall\kX,\kY,\kX_1,\dots,\kX_{n-2}\in\kg^\bullet$ homogeneous, one has:
\begin{equation}
i_\kX(\omega\fois \eta)=(i_\kX \omega)\fois\eta+(-1)^n\eps(|\kX|,|\omega|)\omega(i_\kX\eta),
\end{equation}
due to the definition of $i_\kX$ \eqref{definnerprod} and of the product \eqref{defprod}. Then, the third axiom of \eqref{defcalcdiff} implies:
\begin{align}
i_\kX i_\kY\omega(\kX_1,\dots,\kX_{n-2})&=\eps(|\kX|+|\kY|,|\omega|)\eps(|\kX|,|\kY|)\omega(\kY,\kX,\kX_1,\dots,\kX_{n-2})\nonumber\\
&=-\eps(|\kX|,|\kY|)i_\kY i_\kX\omega(\kX_1,\dots,\kX_{n-2}),
\end{align}
so that $[i_\kX,i_\kY]=0$. Furthermore, $L_\kX\dd=\dd i_\kX\dd=\dd L_\kX$.
\begin{equation}
[L_\kX,i_\kY]=i_\kX\dd i_\kY+\dd i_\kX i_\kY-\eps(|\kX|,|\kY|)(i_\kY i_\kX\dd+i_\kY\dd i_\kX).
\end{equation}
By a long but straightforward calculation, using \eqref{defdiff} and \eqref{definnerprod}, one finds $[L_\kX,i_\kY]=i_{[\kX,\kY]_\eps}$. Finally, the above results permit one to show that
\begin{align}
[L_\kX,L_\kY]&=L_\kX(i_\kY\dd+\dd i_\kY)-\eps(|\kX|,|\kY|)(i_\kY\dd+\dd i_\kY)L_\kX\nonumber\\
&=i_{[\kX,\kY]_\eps}\dd+\dd i_{[\kX,\kY]_\eps}=L_{[\kX,\kY]_\eps}.
\end{align}
\end{proof}

\subsection{\texorpdfstring{$\eps$-connections}{epsilon-connections}}
\label{subsec-epsconn}

Let us now generalize the notion of noncommutative connections \cite{DuboisViolette:1988ir, DuboisViolette:1989vq,DuboisViolette:1998su,DuboisViolette:1999cj, Masson:2007jb} and gauge theory to this framework of $\eps$-graded algebras. Let $\modrM$ be a right module on the $\eps$-graded algebra $\algrA$.

\begin{definition}[$\eps$-connections]
An homogeneous linear map of degree 0,
\begin{equation*}
\nabla:\modrM\to\Omr^{1,\bullet}(\algA,\modM),
\end{equation*}
is called an $\eps$-connection if $\forall a\in\algrA$, $\forall m\in\modrM$,
\begin{equation}
\nabla(ma)=\nabla(m)a+m\dd a.\label{defconn}
\end{equation}
\end{definition}

\begin{proposition}
Let $\nabla$ be an $\eps$-connection on $\modrM$. Then, it can be extended as a linear map
\begin{equation*}
\nabla:\Omr^{p,\bullet}(\algA,\modM)\to\Omr^{p+1,\bullet}(\algA,\modM)
\end{equation*}
using the relation $\forall\omega\in\Omr^{p,|\omega|}(\algA,\modM)$ and $\forall\kX_1,\dots,\kX_{p+1}\in\Der^\bullet_\eps(\algA)$ homogeneous,
\begin{multline}
\nabla(\omega)(\kX_1,\dots,\kX_{p+1})=\sum_{m=1}^{p+1} (-1)^{m+1}f_4\ \nabla(\omega(\kX_1,\dots \omi{m} \dots,\kX_{p+1}))(\kX_m)\label{defnabla}\\
+\sum_{1\leq m<n\leq p+1}(-1)^{m+n}f_5\ \omega([\kX_m,\kX_n]_\eps,\dots \omi{m} \dots \omi{n} \dots,\kX_{p+1}),
\end{multline}
where the factors $f_i$ are given by:
\begin{align}
f_4&=\prod_{a=m+1}^{p+1}\eps(|\kX_m|,|\kX_a|),\nonumber\\
f_5&=\eps(|\kX_n|,|\kX_m|)\prod_{a=1}^{m-1}\eps(|\kX_a|,|\kX_m|)\prod_{a=1}^{n-1}\eps(|\kX_a|,|\kX_n|).\nonumber
\end{align}

Then $\nabla$ satisfies the following relation, $\forall\omega\in\Omr^{p,|\omega|}(\algA,\modM)$, $\forall\eta\in\Omr^{q,|\eta|}(\algA)$ homogeneous,
\begin{equation}
\nabla(\omega\eta)=\nabla(\omega)\eta+(-1)^p\omega\dd\eta.\label{propconn}
\end{equation}

The obstruction for $\nabla$ to be an homomorphism of right $\algrA$-module is measured by its curvature $R=\nabla^2$, homogeneous linear map of degree 0, which takes the following form, $\forall m\in\modrM$, $\forall\kX,\kY\in\Der^\bullet_\eps(\algA)$ homogeneous,
\begin{equation}
R(m)(\kX,\kY)=\eps(|\kX|,|\kY|)\nabla(\nabla(m)(\kY))(\kX)-\nabla(\nabla(m)(\kX))(\kY)-\nabla(m)([\kX,\kY]_\eps).\label{defcurv}
\end{equation}
\end{proposition}

\begin{proof}
The formula \eqref{defnabla} is inspired by the formula \eqref{defdiff} in order that the definition of this extension is well defined. Proving \eqref{propconn} is therefore like proving that $\dd$ is an $\widetilde\eps$-derivation of degree $(1,0)$ in the proposition \ref{propproddiff}, it is a long but straightforward computation.
\end{proof}

\begin{proposition}
The space of all $\eps$-connections on $\modrM$ is an affine space modeled on the vector space $\Hom^0_\algA(\modrM,\Omr^{1,\bullet}(\algA,\modM))$. Furthermore, the curvature $R$ associated to an $\eps$-connection $\nabla$ is a homomorphism of right $\algrA$-modules.
\end{proposition}
\begin{proof}
Let $\nabla$ and $\nabla'$ be two $\eps$-connections, and $\Psi=\nabla-\nabla'$. Then, $\forall a\in\algrA$ and $\forall m\in\modrM$,
\begin{equation}
\Psi(ma)=\nabla(m)a-\nabla'(m)a=\Psi(m)a.
\end{equation}
Therefore, $\Psi\in\Hom^0_\algA(\modrM,\Omr^{1,\bullet}(\algA,\modM))$.

In the same way, from $R=\nabla^2$, using \eqref{defconn} and \eqref{propconn}, one obtains:
\begin{equation}
R(ma)=\nabla^2(m)a-\nabla(m)\dd a+\nabla(m)\dd a+m\dd^2 a=R(m)a.
\end{equation}
\end{proof}

\begin{definition}[Gauge group]
The gauge group of $\modrM$ is defined as the group of automorphisms of degree 0 of $\modrM$ as a right $\algrA$-module, $\Aut_\algA^0(\modM,\modM)$. Its elements are called gauge transformations.

Each gauge transformation $\Phi$ is extended to an automorphism of degree $(0,0)$ of $\Omr^{\bullet,\bullet}(\algA,\modM)$, considered as a right $\Omr^{\bullet,\bullet}(\algA)$-module, in the following way, $\forall\omega\in\Omr^{p,|\omega|}(\algA,\modM)$ and $\forall\kX_1,\dots,\kX_{p}\in\Der^\bullet_\eps(\algA)$,
\begin{equation}
\Phi(\omega)(\kX_1,\dots,\kX_p)=\Phi(\omega(\kX_1,\dots,\kX_p)).
\end{equation}
\end{definition}

\begin{proposition}
The gauge group of $\modrM$ acts on the space of its $\eps$-connections in the following way: for $\Phi\in\Aut_\algA^0(\modM,\modM)$ and $\nabla$ an $\eps$-connection,
\begin{equation}
\nabla^\Phi=\Phi\circ\nabla\circ\Phi^{-1}
\end{equation}
is again an $\eps$-connection. The induced action of $\Phi$ on the associated curvature is given by
\begin{equation}
R^\Phi=\Phi\circ R\circ\Phi^{-1}.
\end{equation}
\end{proposition}

\begin{proof}
The axioms \eqref{defconn} are satisfied for $\nabla^\Phi$: $\forall a\in\algrA$ and $\forall m\in\modrM$,
\begin{align}
\nabla^\Phi(ma)&=\Phi\circ\nabla(\Phi^{-1}(m)a)\nonumber\\
&=(\nabla^\Phi(m)a)+m\dd a.
\end{align}
And since $|\Phi|=|\Phi^{-1}|=0$, $\nabla^\Phi$ is homogeneous of degree 0. The proof is trivial for $R^\Phi=\nabla^\Phi\circ\nabla^\Phi$.
\end{proof}

\bigskip
Let $\algrA$ be an $\eps$-graded involutive ($\gC$-)algebra, and $\modrM$ a right $\algrA$-module with a hermitean structure.

A gauge transformation $\Phi$ is called unitary if $\forall m,n\in\modrM$,
\begin{equation}
\langle\Phi(m),\Phi(n)\rangle=\langle m,n\rangle.
\end{equation}

\begin{proposition}
An $\eps$-connection $\nabla$ is called hermitean if $\forall m,n\in\modrM$ homogeneous and $\forall\kX=\sum_{k\in\Gamma}\kX_k$ a real $\eps$-derivation decomposed in homogeneous components,
\begin{equation}
\sum_{k\in\Gamma}\eps(-|m|+|n|,k)\langle\nabla(m)(\kX_k),n\rangle+\langle m,\nabla(n)(\kX)\rangle=\dd(\langle m,n\rangle)(\kX),
\end{equation}
where $|\kX_k|=k$.

Then, the space of hermitean $\eps$-connections on $\modrM$ is stable under the group of unitary gauge transformations.
\end{proposition}
\begin{proof}
Let $\Phi$ be a unitary gauge transformation, $\nabla$ an hermitean $\eps$-connection, $\kX=\sum_{k\in\Gamma}\kX_k$ a real $\eps$-derivation, and $m,n\in\modrM$. Then,
\begin{align*}
\sum_{k\in\Gamma}&\eps(-|m|+|n|,k)\langle\nabla^\Phi(m)(\kX_k),n\rangle+\langle m,\nabla^\Phi(n)(\kX)\rangle\\
&=\sum_{k\in\Gamma}\eps(-|m|+|n|,k)\langle\nabla\circ\Phi^{-1}(m)(\kX_k),\Phi^{-1}(n)\rangle+\langle\Phi^{-1}(m),\nabla\circ\Phi^{-1}(n)(\kX)\rangle\\
&=\dd(\langle\Phi^{-1}(m),\Phi^{-1}(n)\rangle)(\kX)\\
&=\dd(\langle m,n\rangle)(\kX).
\end{align*}
\end{proof}

\section{\texorpdfstring{Applications to various examples of $\eps$-graded algebras}{Applications to various examples of epsilon-graded algebras}}
\label{sec-ex}

The aim of this section is to illustrate the previous definitions with some results for particular $\eps$-graded algebras. Four typical examples have been chosen here. Firstly, we consider examples of $\eps$-graded commutative algebras and the case of a particular supermanifold. Then, we study two cases of noncommutative $\eps$-graded algebras: matrix algebras with elementary grading, for which the same properties as in the non-graded case \cite{DuboisViolette:1988ir} occur, and matrix algebras with fine grading, very different from the previous case. Finally, we consider a generalization of matrix algebras, the Moyal algebra, from which one can construct a $\gZ_2$-graded algebra, and we give some mathematical explanations about a particular gauge theory on the Moyal space thanks to this superalgebra.

\subsection{\texorpdfstring{$\eps$-graded commutative algebras}{epsilon-graded commutative algebras}}
\label{subsec-excommut}

Firstly, a non-graded associative algebra can be seen as a $\eps$-graded algebra, so that the differential calculus based on the derivations of an associative algebra, presented in \cite{DuboisViolette:1988cr} for instance, is a particular case of the formalism of this chapter. Consequently, for the commutative non-graded algebra $\algA=C^\infty(M)$ of the functions of a smooth compact manifold $M$, the differential calculus of subsection \ref{subsec-diffcalc} is the de Rham complex for the manifold $M$.
\medskip

Let us consider the case of the finite dimensional $\gZ$-graded commutative algebra $\algrA = \exter^\bullet V$ with its usual grading, where $V$ is a vector space of dimension $q$. The commutation factor is taken to be $\eps(p,q) = (-1)^{pq}$ with $p,q \in \gZ$. Denote by $\{\theta_i\}_{i=1,\ldots,q}$ a basis of $V$. Then $\algrA = \exter^\bullet (\theta_1,\dots,\theta_q)$ where in this algebra the $\theta_i$'s are anticommuting variables of degree $1$. The $\eps$-center of this algebra is the whole algebra: $\caZ^\bullet_\eps(\algA) = \algrA$.

Using \eqref{defderiv}, any $\eps$-derivation $\kX$ on $\algrA$ homogeneous of degree $r$ is completely determined by its values on the generators $\theta_i$. Because we are dealing with a graded commutative algebra, one can easily verify that any values $\kX_i = \kX(\theta_i) \in \algA^{r+1}$ are acceptable. Notice the shift in the degrees between the one of $\kX$ as a graded derivation and the degrees of the $\kX_i$'s as elements in $\algrA$. Denote by $\{\alpha^j\}_{j=1,\ldots,q}$ the dual basis of $\{\theta_i\}_{i=1,\ldots,q}$. Then $\alpha^j$ defines a derivation of degree $-1$: $\theta_i \mapsto \delta^j_i \gone$. As a module over the $\eps$-center, one as $\Der^\bullet_\eps(\algA) = \algA^{\bullet + 1} \otimes V^\ast$ where the module structure is the one on $\algrA$ and $V^\ast$ is the dual vector space of $V$. One can explicitly write $\kX = \kX_i \otimes \alpha^i$ with the previous notations.

The structure of $\eps$-Lie algebra of $\Der^\bullet_\eps(\algA)$ can be described as follows. The $\eps$-Lie bracket on $\gone \otimes V^\ast$ is zero and for any $\eps$-derivations $\kX$ and $\kY$ of degrees $r$ and $s$, one has $[\kX, \kY]_\eps = \kX_i \alpha^i(\kY_j) \otimes \alpha^j - (-1)^{rs} \kY_j \alpha^j(\kX_i) \otimes \alpha^i$ with obvious notations. This $\eps$-Lie bracket is a Nijenhuis-Richardson type bracket.

Using the structure of $\Der^\bullet_\eps(\algA)$, any $n$-form $\omega \in \Omr^{n,k}(\algA)$ is completely given by its values on the derivations in $\gone \otimes V^\ast$, \textsl{i.e.} on the $\alpha^i$'s. For instance, $1$-forms of degree $k$ are elements in $\algA^{k} \otimes V$ where we identify $V^{\ast \ast} = V$. For $n$-forms, notice that one has $\omega(\alpha^{i_1}, \ldots, \alpha^{i_p}, \alpha^{i_{p+1}}, \ldots, \alpha^{i_n}) = \omega(\alpha^{i_1}, \ldots, \alpha^{i_{p+1}}, \alpha^{i_{p}}, \ldots, \alpha^{i_n})$ because the $\alpha^i$'s are of odd degree as derivations. The vector space of $n$-forms of degree $k$ is then $\exter^{(k-n)} V \otimes \symes^n V$ where $\symes^\bullet V$ is the symmetric algebra defined on $V$.

In order to be precise, we define $\theta_{i_1} \vee \cdots \vee \theta_{i_n} \in \symes^n V$ as the $n$-form of degree $n$ 
\begin{equation*}
(\theta_{i_1} \vee \cdots \vee \theta_{i_n})(\alpha^{j_1}, \ldots, \alpha^{j_n}) =
(-1)^{\frac{n(n-1)}{2}} \sum_{\sigma \in \kS_n} \theta_{i_1}(\alpha^{j_{\sigma(1)}}) \cdots \theta_{i_n}(\alpha^{j_{\sigma(n)}})
\end{equation*}
With this definition, the product of $2$ forms $\omega = \omega^a \otimes P_a \in \exter^{k} V \otimes \symes^m V$ and $\eta = \eta^b \otimes Q_b \in \exter^{\ell} V \otimes \symes^n V$ is just the product in the graded commutative algebra $\exter^{\bullet} V \otimes \symes^\bullet V$: $\omega \eta = (-1)^{m\ell} \omega^a \wedge \eta^b \otimes P_a \vee Q_b$, so that as a graded commutative algebra $\Omr^{\bullet,\bullet}(\algA) = \exter^{\bullet} V \otimes \symes^\bullet V$.

Applied to the $\alpha^i$'s, the definition of the differential, \eqref{defdiff}, simplifies because the second sum is zero. Its explicit expression on forms as elements of the algebra $\exter^{\bullet} V \otimes \symes^\bullet V$ is then
\begin{multline*}
(\theta_{i_1} \wedge \cdots \wedge \theta_{i_k}) \otimes (\theta_{j_1} \vee \cdots \vee \theta_{j_n}) \mapsto \\
(-1)^n \sum_{\ell = 1}^{k} (-1)^{k-\ell} (\theta_{i_1} \wedge \cdots \wedge \theta_{i_{\ell-1}} \wedge \theta_{i_{\ell+1}} \wedge \cdots \wedge \theta_{i_k}) \otimes (\theta_{i_\ell} \vee \theta_{j_1} \vee \cdots \vee \theta_{j_n})
\end{multline*}
One just ``transfers'' a $\theta_i$ from the antisymmetric part to the symmetric one.

\medskip
Consider now the $\gZ$-graded commutative algebra obtained as a tensor product
\begin{equation*}
\algrA = C^\infty(\gR^{p|q}) = C^\infty(\gR^p) \otimes \exter^\bullet V
\end{equation*}
where as before $V$ is a vector space of dimension $q$. This algebra describes a particular case of supermanifold \cite{Leites:1980}. As in the previous example $\caZ^\bullet_\eps(\algA) = \algrA$.

Denote by $\{\theta_i\}_{i=1,\ldots,q}$ a basis of $V$. Then any element $f \in \algrA$ can be decomposed as
\begin{equation*}
f(x) = \sum_{I\subset\{1,\dots,q\}} f_I(x_1,\dots,x_p) \theta^I
\end{equation*}
where $(x_i)$ is the canonical coordinate system of $\gR^p$, $I$ are ordered subsets of $\{1,\dots,q\}$, and $\theta^I=\wedge_{i\in I}\theta_i$.

One can then show that any $\eps$-derivation on $\algrA$ can be decomposed into two parts: one part acts as a derivation on $C^\infty(\gR^p)$ with values in $C^\infty(\gR^p) \otimes \exter^\bullet V$, and an other part is a smooth function with values in the $\eps$-derivations on $\exter^\bullet V$:
\begin{equation*}
\Der^\bullet_\eps(\algA) = \left[\Gamma(\gR^p) \otimes \exter^\bullet V \right] \oplus \left[ C^\infty(\gR^p) \otimes \exter^{\bullet+1} V \otimes V^\ast \right]
\end{equation*}
where $\Gamma(\gR^p)$ is the usual Lie algebra of vector fields on $\gR^p$ and where we have explicitly used the structure of the space of $\eps$-derivations on $\exter^\bullet V$.

As a $\caZ^\bullet_\eps(\algA)$-module, $\Der^\bullet_\eps(\algA)$ is generated by the two disjoint ($\eps$-)Lie algebras $\Gamma(\gR^p)$ and $V^\ast$, so that
\begin{equation*}
\underline\Omega_\eps^{\bullet,\bullet}(\algA) = \Omega^\bullet_{\text{dR}}(\gR^p) \otimes \Omega^{\bullet,\bullet}(\exter^\bullet V)
\end{equation*}
as graded commutative differential algebras.

\subsection{\texorpdfstring{$\eps$-graded matrix algebras with elementary grading}{epsilon-graded matrix algebras with elementary grading}}
\label{subsec-exelem}

Let $\Gamma$ be an abelian group and $\algA=\Matr_D$ the algebra of $D\times D$ complex matrices, such that $\algA$ is a $\Gamma$-graded algebra:
\begin{equation}
\algrA=\bigoplus_{\alpha\in\Gamma}\algA^\alpha.\label{gradingmatrix}
\end{equation}
Let $(E_{ij})_{1 \leq i, j \leq D}$ be the canonical basis of $\algrA$, whose product is as usual
\begin{equation}
E_{ij}\fois E_{kl}=\delta_{jk}E_{il}\label{prodmat}.
\end{equation}
\begin{definition}[Elementary grading \cite{Bahturin:2002}]
The grading \eqref{gradingmatrix} is called elementary if there exists a map $\varphi:\{1,\dots,D\}\to\Gamma$ such that $\forall i,j\in\{1,\dots,D\}$, $E_{ij}$ is homogeneous of degree $|E_{ij}|=\varphi(i)-\varphi(j)$.
\end{definition}

From now on, we suppose that the grading \eqref{gradingmatrix} is elementary. Then, one can check that the usual conjugation is an involution for $\algrA$. Furthermore, $\algrA$ can be characterized by the following result:
\begin{proposition}[\cite{Bahturin:2002}]
The matrix algebra $\algrA=\Matr_D$, with an elementary $\Gamma$-grading, is isomorphic, as a $\Gamma$-graded algebra, to the endomorphism algebra of some $D$-dimensional $\Gamma$-graded vector space.
\end{proposition}
More general gradings on $\algA$ have been classified in \cite{Bahturin:2002}.

\subsubsection{Properties}
\label{subsub-propelem}

Let $\eps:\Gamma\times\Gamma\to\gC^\ast$ be a commutation factor, which turns the elementary $\Gamma$-graded involutive algebra $\algrA$ into an $\eps$-graded algebra.

\begin{proposition}
\label{prop-centermatrix}
The $\eps$-center of $\algrA$ is trivial:
\begin{equation}
\caZ_\eps^\bullet(\algA)=\caZ_\eps^0(\algA)=\gC\gone.
\end{equation}
\end{proposition}

\begin{proof}
Let $A\in\caZ_\eps^\bullet(\algrA)$ written as $A=\sum_{i,j=1}^Da_{ij}E_{ij}$. Then, due to \eqref{prodmat}, we immediately get, $\forall k,l\in\{1,\dots,D\}$,
\begin{equation}
0=[A,E_{kl}]_\eps=\sum_{i,j=1}^D(a_{ik}\delta_{jl}-\eps(\varphi(l)-\varphi(j),\varphi(k)-\varphi(l))a_{lj}\delta_{ik})E_{ij}.
\end{equation}
Therefore, $\forall i,j,k,l$,
\begin{equation}
a_{ik}\delta_{jl}=\eps(\varphi(l)-\varphi(j),\varphi(k)-\varphi(l))a_{lj}\delta_{ik}.
\end{equation}
For $i\neq k$ and $j=l$, we get $a_{ik}=0$, and for $i=k$ and $j=l$, $a_{ii}=a_{jj}$. This means that $A\in\gC\gone$.
\end{proof}

\begin{proposition}[$\eps$-traces]
\label{prop-epstracematrixalgebra}
For any $A=(a_{ij})\in\algrA$, the expression
\begin{equation}
\tr_\eps(A)=\sum_{i=1}^D\eps(\varphi(i),\varphi(i))a_{ii}
\end{equation}
defines a real $\eps$-trace on $\algrA$. Moreover, the space of $\eps$-traces on $\algrA$ is one-dimensional.
\end{proposition}

\begin{proof}
Let $T$ be an $\eps$-trace on $\algrA$. From \eqref{deftrace}, we get, $\forall i,j,k,l$, $T(E_{ij}\fois E_{kl})=\eps(\varphi(i)-\varphi(j),\varphi(k)-\varphi(l))T(E_{kl}\fois E_{ij})$. Using \eqref{prodmat}, one can obtain
\begin{equation}
\delta_{jk}T(E_{il})=\eps(\varphi(i)-\varphi(j),\varphi(k)-\varphi(l))\delta_{il}T(E_{kj}).
\end{equation}
Then, with $i\neq l$ and $j=k$, one gets $T(E_{il})=0$. On the other hand, with $i=l$ and $j=k$, one has $T(E_{ii})=\eps(\varphi(i)-\varphi(j),\varphi(j)-\varphi(i))T(E_{jj}) =\eps(\varphi(i),\varphi(i))\eps(\varphi(j),\varphi(j))T(E_{jj})$. So $\tr_\eps$ is an $\eps$-trace on $\algrA$ and there exists $\lambda\in\gC$ such that $T=\lambda\tr_\eps$.
\end{proof}

In the following, we will denote by $\ksl_\eps^\bullet(D)$ the $\eps$-Lie subalgebra of $\algrA$ of $\eps$-traceless elements.

\begin{proposition}
\label{prop-derivmatrix}
All the $\eps$-derivations of $\algrA$ are inner:
\begin{equation}
\Out_\eps^\bullet(\algA)=\algzero.
\end{equation}
Moreover, the space of real $\eps$-derivations of $\algrA$ is isomorphic to the space of antihermitean matrices.
\end{proposition}

\begin{proof}
Let $\kX\in\Der_\eps^\bullet(\algA)$ homogeneous, decomposed on the basis as $\kX(E_{ij})=\sum_{k,l=1}^D\kX_{ij}^{kl}E_{kl}$, with $\kX_{ij}^{kl}=0$ if $\varphi(k)-\varphi(l)-\varphi(i)+\varphi(j)\neq|\kX|$ where $|\kX|$ is the degree of $\kX$. Then, due to \eqref{defderiv}, we have, $\forall i,j,k,l$,
\begin{equation}
\kX(E_{ij}\fois E_{kl})=\kX(E_{ij})\fois E_{kl}+\eps(|\kX|,\varphi(i)-\varphi(j))E_{ij}\fois \kX(E_{kl}).\label{deriv1}
\end{equation}
Using \eqref{prodmat}, this can be written as
\begin{equation}
\sum_{a,b=1}^D\delta_{jk}\kX_{il}^{ab}E_{ab}=\sum_{a,b=1}^D(\delta_l^b\kX_{ij}^{ak}+\eps(|\kX|,\varphi(i)-\varphi(j))\delta_i^a\kX_{kl}^{jb})E_{ab}.
\end{equation}
For $b=l$, and by changing some indices, we obtain, $\forall i,j,k,l,a$,
\begin{equation}
\kX_{ij}^{kl}=\delta_{jl}\kX_{ia}^{ka}-\eps(|\kX|,\varphi(i)-\varphi(j))\delta_i^k\kX_{la}^{ja}.\label{deriv2}
\end{equation}

On the other hand, let us define $M_\kX=\sum_{k=1}^D\kX(E_{ka})\fois E_{ak}=\sum_{k,l=1}^D\kX_{ka}^{la}E_{lk} \in \algrA$ for an arbitrary $a$. Then, using \eqref{prodmat}, we find
\begin{equation}
[M_\kX,E_{ij}]_\eps=\sum_{k,l=1}^D(\delta_j^l\kX_{ia}^{ka}-\eps(\varphi(j)-\varphi(l),\varphi(i)-\varphi(j))\delta_i^k\kX_{la}^{ja})E_{kl}.
\end{equation}
Since $\kX_{la}^{ja}=0$ for $\varphi(j)-\varphi(l)\neq|\kX|$, \eqref{deriv2} implies that $[M_\kX,E_{ij}]_\eps=\sum_{k,l=1}^D\kX_{ij}^{kl}E_{kl}=\kX(E_{ij})$, and $\kX$ is an inner derivation generated by $M_\kX$.

The statement about reality can be checked easily.
\end{proof}

\subsubsection{\texorpdfstring{Differential calculus and $\eps$-connections}{Differential calculus and epsilon-connections}}
\label{subsub-diffcalcelem}

In this subsection, we describe the differential calculus based on $\eps$-derivations for a certain class of $\eps$-graded matrix algebras with elementary grading. In order to do that, we introduce the following algebra, which is the equivalent of the exterior algebra in the present framework.

\begin{definition}[$\eps$-exterior algebra]
Let $V^\bullet$ be a $\Gamma$-graded vector space and $\eps$ a commutation factor on $\Gamma$. One defines the $\eps$-exterior algebra on $V^\bullet$, denoted by $\exter_{\eps}^\bullet V^\bullet$, as the tensor algebra of $V^\bullet$ quotiented by the ideal generated by
\begin{equation*}
\{\kX\otimes\kY+\eps(|\kX|,|\kY|)\kY\otimes\kX,\,\,\kX,\kY\in V^\bullet\,\text{homogeneous}\}.
\end{equation*}
\end{definition}

\begin{proposition}[Structure of $\exter_{\eps}^\bullet V^\bullet$]
\label{prop-structureexterioralgebra}
The $(\gZ\times\Gamma)$-graded algebra $\exter_{\eps}^\bullet V^\bullet$ is $\widetilde{\eps}$-graded commutative for the commutation factor $\widetilde{\eps}((m, i), (n, j)) = (-1)^{mn} \eps(i,j)$.

One has the following factor decomposition as tensor products of $\widetilde{\eps}$-graded commutative algebras:
\begin{equation}
\exter_{{}_\eps}^\bullet V^\bullet= \Big(\bigotimes_{\substack{\alpha\in\Gamma\\ \eps(\alpha,\alpha)=1}}\exter^\bullet V^\alpha\Big) \otimes \Big(\bigotimes_{\substack{\alpha\in\Gamma\\ \eps(\alpha,\alpha)=-1}}\symes^\bullet V^\alpha\Big),
\end{equation}
where $\exter^\bullet V^\alpha$ is the (usual) exterior algebra of the vector space $V^\alpha$ and $\symes^\bullet V^\alpha$ is the (usual) symmetric algebra.
\end{proposition}

The proof is just a straighforward adaptation of the one for the usual exterior algebra. Nevertheless, notice that the factors are of two types: an exterior algebra or a symmetric algebra.

\bigskip
Let $\eps$ be a commutation factor on an abelian group $\Gamma$, $\kg^\bullet$ an $\eps$-Lie algebra for this commutation factor and $V^\bullet$ a $\Gamma$-graded vector space of representation of $\kg^\bullet$. Let us introduce the $(\gZ\times\Gamma)$-graded vector space $\Omega_\eps^{\bullet,\bullet}(\kg,V)=V^\bullet\otimes \exter_{\eps}^\bullet(\kg^\ast)^\bullet$.

\begin{proposition}
\label{prop-gv}
$\Omega_\eps^{\bullet,\bullet}(\kg,V)$ is a $(\gZ\times\Gamma)$-graded differential complex for the differential of degree $(1,0)$ defined by
\begin{multline}
\label{defdiff2}
\dd\omega(\kX_1,\dots,\kX_{p+1})=\sum_{m=1}^{p+1} (-1)^{m+1}g_1\ \kX_m\omega(\kX_1,\dots \omi{m} \dots,\kX_{p+1})\\
+\sum_{1\leq m<n\leq p+1}(-1)^{m+n}g_2\ \omega([\kX_m,\kX_n]_\eps,\dots \omi{m} \dots \omi{n} \dots,\kX_{p+1}),\end{multline}
$\forall\omega\in\Omega_\eps^{p,|\omega|}(\kg,V)$ and $\forall\kX_1,\dots,\kX_{p+1}\in\kg^\bullet$ homogeneous, where
\begin{align*}
g_1&=\eps(|\omega|,|\kX_m|)\prod_{a=1}^{m-1}\eps(|\kX_a|,|\kX_m|)\\
g_2&=\eps(|\kX_n|,|\kX_m|)\prod_{a=1}^{m-1}\eps(|\kX_a|,|\kX_m|)\prod_{a=1}^{n-1}\eps(|\kX_a|,|\kX_n|).
\end{align*}

Moreover, in the case where $V^\bullet$ is an $\eps$-graded algebra and $\kg^\bullet = V^\bullet_{\Lie,\eps}$ is its associated $\eps$-Lie algebra, acting by the adjoint representation on $V^\bullet$, $\Omega_\eps^{\bullet,\bullet}(\kg,V)$ is a $(\gZ\times\Gamma)$-graded differential algebra for the product,
\begin{equation}
\label{defprod2}
(\omega\fois\eta)(\kX_1,\dots,\kX_{p+q}) =\frac{1}{p!q!}\sum_{\sigma\in\kS_{p+q}}(-1)^{|\sigma|}g_3\ \omega(\kX_{\sigma(1)},\dots,\kX_{\sigma(p)})\fois\eta(\kX_{\sigma(p+1)},\dots,\kX_{\sigma(p+q)}),
\end{equation}
$\forall\omega\in\Omega_\eps^{p,|\omega|}(\kg,V)$, $\forall\eta\in\Omega_\eps^{q,|\eta|}(\kg,V)$, and $\forall\kX_1,\dots,\kX_{p+q}\in\kg^\bullet$ homogeneous, where
\begin{equation*}
g_3 = \prod_{m<n,\sigma(m)>\sigma(n)}\eps(|\kX_{\sigma(n)}|,|\kX_{\sigma(m)}|)\prod_{m\leq p}\eps(|\eta|,|\kX_{\sigma(m)}|).
\end{equation*}
\end{proposition}

\begin{proof}
The product \eqref{defprod2} and the differential \eqref{defdiff2} are formally the same as the product \eqref{defprod} and the differential \eqref{defdiff}.
\end{proof}

This permits one to show the following Theorem, which gives the structure of the $\eps$-derivation-based differential calculus on a certain class of $\eps$-graded matrix algebras for elementary gradings.

\begin{theorem}[The $\eps$-derivation-based differential calculus]
\label{thm-calcdiff}
Let $\algrA=\Matr_D$ be an $\eps$-graded matrix algebra with elementary grading, and suppose that:
\begin{equation}
\tr_\eps(\gone)=\sum_{i=1}^D \eps(\varphi(i),\varphi(i)) \neq0,\label{condsimpl}
\end{equation}
Then:
\begin{itemize}
\item The adjoint representation $\ad:\ksl_\eps^\bullet(D)\to\Der_\eps^\bullet(\algA)$ is an isomorphism of $\eps$-Lie algebras.
\item Let $\ad_\ast:\Omega_\eps^{\bullet,\bullet}(\ksl_\eps(D),\algA)\to\Omr^{\bullet,\bullet}(\algA)$ be the push-forward of $\ad$ and let $(E_i)$ be a basis of $\algrA$. Then, for any $\omega\in\Omega_\eps^{k,\bullet}(\ksl_\eps(D),\algA)$, one has the relation:
\begin{equation}
\ad_\ast(\omega)(\ad_{E_{i_1}},\dots,\ad_{E_{i_k}})=\omega(E_{i_1},\dots,E_{i_k}).\label{exprbasead}
\end{equation}
\item $\ad_\ast$ is an isomorphism of $(\gZ\times\Gamma)$-graded differential algebras.
\end{itemize}
\end{theorem}

\begin{proof}
Since $\ksl_\eps(D)\cap\caZ_\eps(\algA)=\algzero$, by Proposition~\ref{prop-centermatrix} and using the condition~\eqref{condsimpl}, the kernel of the surjective map $\ad:\ksl_\eps^\bullet(D)\to\Der_\eps^\bullet(\algA)=\Int_\eps^\bullet(\algA)$ is trivial, so that it is an isomorphism of vector spaces. Due to the $\eps$-Jacobi identity \eqref{epsjacobi}, one has $\forall a,b\in\algrA$, $\ad_{[a,b]_\eps}=[\ad_a,\ad_b]_\eps$, which proves that $\ad:\ksl_\eps^\bullet(D)\to\Der_\eps^\bullet(\algA)$ is an isomorphism of $\eps$-Lie algebras.

Equation~\eqref{exprbasead} is then a trivial consequence of this isomorphism.

Let us stress that $\ksl_\eps(D)$ is only a subalgebra of the associated $\eps$-Lie algebra of $\algrA$, but the result of Proposition~\ref{prop-gv} generalizes to this case, so that $\Omega_\eps^{\bullet,\bullet}(\ksl_\eps(D),\algA)$ is a $(\gZ\times\Gamma)$-graded differential algebra. It is straightforward to see that $\ad_\ast$ is a morphism of $(\gZ\times\Gamma)$-graded differential algebras. Indeed, the products and the differentials have the same formal definitions, respectively \eqref{defprod} and \eqref{defprod2}, \eqref{defdiff} and \eqref{defdiff2}, for the two complexes $\Omr^{\bullet,\bullet}(\algA)$ and $\Omega_\eps^{\bullet,\bullet}(\ksl_\eps(D),\algA)$. Equation~\eqref{exprbasead} shows that $\ad_\ast$ is injective. For $(E_i)$ a basis of $\algA$, adapted to the decomposition $\algrA=\gC\gone\oplus\ksl_\eps^\bullet(D)$ (with $E_0=\gone$), $(\theta^i)$ its dual basis, and $\eta\in\Omr^{k,\bullet}(\algA)$, we set $\eta(\ad_{E_{i_1}},\dots,\ad_{E_{i_k}})=\eta^j_{i_1,\dots,i_k}E_j\in\algrA$. Then $\eta=\ad_\ast(\eta^j_{i_1,\dots,i_k}E_j\theta^{i_1}\wedge\dots\wedge\theta^{i_k})$, with $\eta^j_{i_1,\dots,i_k}=0$ if one of $i_1,\dots,i_k$ is zero. Therefore, $\ad_\ast$ is surjective, and it is an isomorphism of bigraded differential algebras.
\end{proof}

Let us notice that there exist some examples of $\eps$-graded matrix algebras for which the condition \eqref{condsimpl} is not satisfied (see Example~\ref{ex-supermatrix} below with $m=n$).

\bigskip
Let us now describe explicitely the space of $\eps$-connections.
\begin{proposition}
The space of $\eps$-connections on $\algrA$, considered as a module over itself, is an affine space modeled on the vector space
\begin{equation}
\Omr^{1,0}(\algA)=\bigoplus _{\alpha\in\Gamma}
\algA^\alpha\otimes\ksl_\eps^\alpha(D)^\ast,
\end{equation}
and containing the (trivial) $\eps$-connection $\dd$.

Moreover, $\ad^{-1}$ can be seen as a $1$-form, and $\dd+\ad^{-1}$ defines a gauge invariant $\eps$-connection.
\end{proposition}

\begin{proof}
Let $\nabla$ be an $\eps$-connection. Define $\omega(\kX)=\nabla(\gone)(\kX)$, $\forall\kX\in\Der_\eps^\bullet(\algA)$. Then, $\omega\in\Omr^{1,0}(\algA)$ and, $\forall a\in\algrA$, $\nabla(a)=\dd a+\omega\fois a$.

On the other hand, $\ad^{-1}:\Der_\eps^\bullet(\algA)\to\ksl_\eps^\bullet(D)\subset\algrA$ is a $1$-form of degree $0$. Since all $\eps$-derivations of $\algrA$ are inner, one has, $\forall \kX\in\Der_\eps^\bullet(\algA)$, $\kX(a)=[\ad^{-1}(\kX),a]_\eps$. Then, it is straightforward to prove the gauge invariance of the connection $\dd+\ad^{-1}$.
\end{proof}

The noncommutative $1$-form $\ad^{-1}$ defined here is the exact analog of the noncommutative $1$-form $i\theta$ defined in the context of the derivation-based differential calculus on the matrix algebra, and it gives rise also to a gauge invariant connection (see \cite{DuboisViolette:1988cr,DuboisViolette:1988ir,DuboisViolette:1998su, Masson:1999ea} for details). This shows that some of the results obtained in the non-graded case remain valid in this new context. But one has to be careful that the condition~\eqref{condsimpl} has to be fulfilled.

\subsubsection{Concrete examples}
\label{subsub-concretelem}

In this part, we give concrete example of $\eps$-graded matrix
algebras in the case when the abelian group $\Gamma$ is freely
generated by a finite number of generators $\{e_r\}_{r\in I}$, and
when $\eps$ is a commutation factor on $\Gamma$ over $\gK=\gC$

We first recall the definitions given in \cite{Rittenberg:1978mr} of color algebras and superalgebras. If, $\forall r\in I$, $\eps(e_r,e_r)=1$, then the $\eps$-graded algebra $\algrA$ is called a color algebra. Otherwise, it is a color superalgebra.

In the following examples, we will consider three gradings: the trivial case $\Gamma=\algzero$, the usual case $\Gamma=\gZ_2$, and a third case $\Gamma=\gZ_2\times\gZ_2$. Proposition~\ref{prop-factcomm} determines the possible commutation factors for theses groups, and they have been given in Example~\ref{ex-factcomm} for the two latter groups.

\begin{example}[$\Gamma=\algzero$]
\label{ex-nogrmatrix}
We consider here the trivial grading $\Gamma=\algzero$ on the matrix algebra $\algA=\Matr(n)=\Matr_n$, so that the commutation factor is also trivial: $\eps(i,j)=1$. The $\eps$-commutator $[-,-]_\eps$ and the trace $\tr_\eps$ are the usual non-graded ones for matrices. The $\eps$-center and the $\eps$-derivations of this $\eps$-graded algebra are given by: $\caZ_\eps(\algA)=\gC\gone$ and $\Der_\eps(\algA)=\Int_\eps(\algA)=\ksl_n$, the usual Lie algebra of traceless matrices. The $\eps$-derivation-based differential calculus coincides with the (usual) derivation-based differential calculus studied in \cite{DuboisViolette:1988cr, DuboisViolette:1988ir, Masson:1995ph}:
\begin{equation}
\underline{\Omega}^\bullet_{\eps}(\Matr(n))=\underline{\Omega}^\bullet_{\Der}(\Matr(n))\approx\Matr(n)\otimes\Big(\exter^\bullet{\ksl_n}^\ast\Big)
\end{equation}
It is finite-dimensional and its cohomology is:
\begin{align}
H^\bullet(\Omr(\Matr(n)),\dd)=H^\bullet(\ksl_n)=\mathcal{I}(\exter^\bullet\ksl_n^\ast),
\end{align}
the algebra of invariant elements for the natural Lie derivative.
\end{example}

\begin{example}[$\Gamma=\gZ_2$]
\label{ex-supermatrix}
Consider now the case $\Gamma=\gZ_2$, with the usual commutation factor $\eps(i,j)=(-1)^{ij}$, and the matrix superalgebra $\algrA=\Matr(m,n)$. This superalgebra is represented by $(m+n)\times(m+n)$ matrices:
\begin{align}
M=\begin{pmatrix} M_{11}&M_{12} \\ M_{21}& M_{22} \end{pmatrix}\in\algrA,
\end{align}
where $M_{11}$, $M_{12}$, $M_{21}$ and $M_{22}$ are respectively $m\times m$, $m\times n$, $n\times m$ and $n\times n$ (complex) matrices. The $\gZ_2$-grading is defined such that $M_{11}$ and $M_{22}$ correspond to degree $0\in\gZ_2$, whereas $M_{12}$ and $M_{21}$ are in degree $1\in\gZ_2$, so that this grading is elementary. Using Proposition~\ref{prop-epstracematrixalgebra}, we find that
\begin{align}
\tr_\eps(M)=\tr(M_{11})-\tr(M_{22}),
\end{align}
is an $\eps$-trace, $\caZ_\eps^\bullet(\algA)=\gC\gone$ and $\Der_\eps^\bullet(\algA)=\Int_\eps^\bullet(\algA)$.

Notice that when $m=n$, one has $\tr_\eps(\gone)=0$, so that condition \eqref{condsimpl} is not satisfied.

If $m \neq n$, one can suppose, by convention, that $m>n$. In that case, $\tr_\eps(\gone)\neq0$, and using Theorem~\ref{thm-calcdiff}, one gets
\begin{equation*}
\Der_\eps^\bullet(\algA)=\ksl^\bullet_\eps(m,n)=\ksl^0_\eps(m,n)\oplus\ksl^1_\eps(m,n),
\end{equation*}
and the associated differential calculus based on these superderivations is given by:
\begin{equation}
\Omr^{\bullet,\bullet}(\Matr(m,n))\approx\Matr^{\bullet}(m,n)\otimes\Big(\exter^\bullet\ksl^0_\eps(m,n)^\ast\Big) \otimes \Big(\symes^\bullet\ksl^1_\eps(m,n)^\ast\Big).\label{calcdiffsupermatrix}
\end{equation}
In this decomposition, the second tensor product is the one of $\widetilde{\eps}$-graded algebras as in Proposition~\ref{prop-structureexterioralgebra}.
Note that this expression involves the symmetric algebra of the odd part of $\ksl^\bullet_\eps(m,n)^\ast$, which means that $\Omr^{\bullet,\bullet}(\Matr(m,n))$ is infinite dimensional as soon as $n>0$ (remember that $m>n$), even if $\ksl_\eps^\bullet(m,n)$ is finite dimensional. This is a key difference with the non-graded case (see Example~\ref{ex-nogrmatrix}). The cohomology of \eqref{calcdiffsupermatrix} has been computed in \cite{Grosse:1999} and is given by:
\begin{align}
H^{\bullet,\bullet}(\Omr(\Matr(m,n)),\dd)=H^{\bullet,0}(\Omr(\Matr(m,n)),\dd)=H^\bullet(\ksl_m),
\end{align}
This is exactly the cohomology of the (non-graded) Lie algebra $\ksl_m$.
\end{example}

\begin{example}[$\Gamma=\gZ_2\times\gZ_2$]
\label{ex-coloralgebra}
Let us now consider the case of $\Gamma=\gZ_2\times\gZ_2$ gradings, with the commutation factor $\eps(i,j)=(-1)^{i_1j_2+i_2j_1}$. Let $\algrA=\Matr(m,n,r,s)$ be the $\eps$-graded algebra of $(m+n+r+s)\times(m+n+r+s)$ matrices defined as follow: any element in $\algrA$ is written as:
\begin{equation}
M=\begin{pmatrix} M_{11}&M_{12}&M_{13}&M_{14} \\ M_{21}& M_{22}&M_{23}&M_{24} \\ M_{31}&M_{32}&M_{33}&M_{34} \\ M_{41}&M_{42}&M_{43}&M_{44} \end{pmatrix}\in\algrA,\label{colormatrix}
\end{equation}
where $M_{ij}$ are rectangular matrices. The grading is such that $M_{11}$, $M_{22}$, $M_{33}$ and $M_{44}$ correspond to degree $(0,0)\in\Gamma$; $M_{12}$, $M_{21}$, $M_{34}$ and $M_{43}$ correspond to degree $(1,0)$; $M_{13}$, $M_{24}$, $M_{31}$ and $M_{42}$ correspond to degree $(0,1)$; $M_{14}$, $M_{23}$, $M_{32}$ and $M_{41}$ correspond to degree $(1,1)$. This is an elementary grading, $\algrA$ is a color algebra in the sense of \cite{Rittenberg:1978mr}, and is therefore a less trivial example of $\eps$-graded algebra than the usual matrix algebra or the super matrix algebra described in Examples~\ref{ex-nogrmatrix} and \ref{ex-supermatrix}.

\begin{equation}
\tr_\eps(M)=\tr(M_{11})+\tr(M_{22})+\tr(M_{33})+\tr(M_{44})=\tr(M),
\end{equation}
is an $\eps$-trace, and $\caZ_\eps^\bullet(\algA)=\gC\gone$ and
$\Der_\eps^\bullet(\algA)=\Int_\eps^\bullet(\algA)=\ksl^\bullet_\eps(m,n,r,s)$.
Moreover, the differential calculus based on $\eps$-derivations
is:
\begin{equation}
\Omr^{\bullet,\bullet}(\Matr(m,n,r,s))=\Matr(m,n,r,s)\otimes
\Big(\exter_\eps^\bullet\ksl_\eps(m,n,r,s)^\ast\Big)\label{expcalcdiffelem}
\end{equation}
\end{example}

\begin{example}[$\Gamma=\gZ_2\times\gZ_2$]
Consider the same grading group $\Gamma$ on the same algebra $\algrA$ as in Example~\ref{ex-coloralgebra}, but with a different commutation factor: $\eps(i,j)=(-1)^{i_1j_1+i_2j_2}$. $\algrA$ is then a color superalgebra but not a color algebra (here $\eps((1,0),(1,0))=-1$). General results lead us to the $\eps$-trace,
\begin{equation}
\tr_\eps(M)=\tr(M_{11})-\tr(M_{22})-\tr(M_{33})+\tr(M_{44}),
\end{equation}
and one has $\caZ_\eps^\bullet(\algA)=\gC\gone$,
$\Der_\eps^\bullet(\algA)=\Int_\eps^\bullet(\algA)$. If $m+s\neq
n+r$, one gets
$\Der_\eps^\bullet(\algA)=\ksl^\bullet_\eps(m,n,r,s)$. The
differential calculus is also described by equation
\eqref{expcalcdiffelem}, but $\ksl^\bullet_\eps(m,n,r,s)$ is
different in this example from Example \ref{ex-coloralgebra}.
\end{example}

The explicit computation of the commutators $[-,-]_\eps$ for the $\eps$-Lie algebras of $\eps$-derivations are not given here because they give rise to cumbersome expressions. Nevertheless, let us mention that they are different for the four above cases, and therefore the $\eps$-Lie algebras of $\eps$-derivations, $\Der_\eps^\bullet(\algA)$, are different for these four examples.

\subsection{\texorpdfstring{$\eps$-graded matrix algebras with fine grading}{epsilon-graded matrix algebras with fine grading}}
\label{subsec-exfine}

In this subsection, we study the case of a fine-grading for the matrix algebra $\algrA$.

\begin{definition}[Fine grading]
Let $\algrA=\Matr_D$ be the complex matrix algebra, graded by an abelian group $\Gamma$.

The grading \eqref{gradingmatrix} of $\algrA$ is called fine if $\forall\alpha\in\Gamma$, $\dim_\gC(\algA^\alpha)\leq 1$. Then, we define the support of the grading:
\begin{equation*}
 \Supp(\algrA)=\{\alpha\in\Gamma,\, \algA^\alpha\neq\algzero\}.
\end{equation*}
\end{definition}

Let $\Gamma$ be an abelian group, and $\algA=\Matr_D$ the algebra of $D\times D$ complex matrices, such that $\algrA$ is a fine $\Gamma$-graded algebra. Let $(e_\alpha)_{\alpha\in\Supp(\algrA)}$ be a homogeneous basis of $\algrA$.

\begin{proposition}
With the above hypotheses on $\algrA$, $\algrA$ is a graded division algebra, namely all non-zero homogeneous elements of $\algrA$ are invertible. Moreover, $\Supp(\algrA)$ is a subgroup of $\Gamma$.

The fine grading of $\algrA$ is determined by the choice of the basis $(e_\alpha)$ and the factor set $\sigma:\Supp(\algrA)\times\Supp(\algrA)\to\gC^\ast$ defined by: $\forall\alpha,\beta\in\Supp(\algrA)$,
\begin{equation*}
e_\alpha\fois e_\beta=\sigma(\alpha,\beta)e_{\alpha+\beta}.
\end{equation*}
Furthermore, $\algrA$ is a fine $\eps_\sigma$-graded commutative algebra, where the commutation factor $\eps_\sigma$ is defined by \eqref{eq-multcomm}.
\end{proposition}
\begin{proof}
The proof of the first property is in \cite{Bahturin:2002}. It uses the fact that $\algrA$, as an associative algebra, does not contain any proper ideal.

$\forall\alpha,\beta\in\Supp(\algrA)$, $e_\alpha\fois e_\beta$ is proportional to $e_{\alpha+\beta}$ and is different from $0$ because of the above property. This defines $\sigma(\alpha,\beta)$. Then, $\sigma$ is a factor set (see Definition \ref{def-multiplier}) since $\algrA$ is associative. For $\alpha,\beta\in\Supp(\algrA)$, one has:
\begin{equation*}
e_\alpha\fois e_\beta=\sigma(\alpha,\beta)e_{\alpha+\beta}=\eps_\sigma(\alpha,\beta)e_\beta\fois e_\alpha,
\end{equation*}
with $\eps_\sigma(\alpha,\beta)=\sigma(\alpha,\beta)\sigma(\beta,\alpha)^{-1}$ (see \eqref{eq-multcomm}).
\end{proof}
Note that $e_0=\sigma(0,0)\gone$.

Let us give a typical example of fine-graded matrix algebra: Clifford algebra.
\begin{definition}
Let $\Gamma$ be an abelian group, and $\sigma$ a factor set of $\Gamma$. We define here $S=\gC\rtimes_\sigma\Gamma$, the {\it crossed-product} of $\gC$ by $\Gamma$ relatively to $\sigma$. $S$ is the algebra of functions $\Gamma\to\gC$ which vanishes outside of a finite number of elements of $\Gamma$, with product: $\forall f,g\in S$, $\forall k\in\Gamma$,
\begin{equation*}
(f\fois g)(k)=\sum_{i+j=k}\sigma(i,j)f(i)g(j).
\end{equation*}
Let $(e_k)_{k\in\Gamma}$ be its canonical basis, given by $\forall i,k\in\Gamma$, $e_k(i)=\delta_{ik}$. It satisfies $\forall i,j\in\Gamma$, $e_i\fois e_j=\sigma(i,j)e_{i+j}$. With the natural fine $\Gamma$-grading given by $S^k=\gC e_k$, $S^\bullet$ is an $\eps_\sigma$-graded commutative algebra.
\end{definition}

\begin{example}
Let $\Gamma=(\gZ_2)^n$, for $n\in\gN^\ast$, and $\sigma$ the factor set of $\Gamma$ defined by: $\forall i,j\in\Gamma$,
\begin{equation*}
\sigma(i,j)=(-1)^{\sum_{1\leq p<q\leq n}i_p j_q}.
\end{equation*}
Then, $\gC\rtimes_\sigma(\gZ_2)^n$ is isomorphic to the Clifford algebra $\mathcal{C} l(n,\gC)$. This is still true for any factor set equivalent to $\sigma$.
\end{example}

\subsubsection{Properties}
\label{subsub-propfine}

The algebra $\algrA=\Matr_D$ has therefore a natural commutation factor $\eps_\sigma$, and one can ask about the properties of $\algrA$ if it is endowed with another general commutation factor $\eps$ on $\Gamma$. We have to ``compare'' $\eps$ with $\eps_\sigma$. In this subsection, $\eps$ will denote a commutation factor on $\Gamma$ over $\gC$ (potentially different from $\eps_\sigma$) satisfying the hermitean condition \eqref{eq-herm-factcom}.

\begin{proposition}
Let $\eps_1$ and $\eps_2$ be two commutations factors on $\Gamma$ over $\gC$. We denote
\begin{equation*}
\Gamma_{\eps_1,\eps_2}=\{i\in\Gamma,\,\forall j\in\Gamma,\,\eps_1(i,j)=\eps_2(i,j)\}.
\end{equation*}
$\Gamma_{\eps_1,\eps_2}$ is a subgroup of $\Gamma$ compatible with the signature decomposition, that is: $\forall i\in\gZ_2$,
\begin{equation*}
\Gamma_{\eps_1,\eps_2}\cap\Gamma_{\eps_1}^i=\Gamma_{\eps_1,\eps_2}\cap\Gamma_{\eps_2}^i.
\end{equation*}
\end{proposition}

\begin{proposition}[$\eps$-center]
\label{prop-centfine}
The $\eps$-center of $\algrA$, whose fine grading is associated to the factor set $\sigma$, is given by:
\begin{equation*}
\caZ_\eps^\bullet(\algA)=\bigoplus_{\alpha\in\Gamma_{\eps,\eps_\sigma}}\algA^\alpha.
\end{equation*}
\end{proposition}
\begin{proof}
Suppose that there exists $\alpha\in\Supp(\algrA)$ such that $e_\alpha\in\caZ^\bullet_\eps(\algA)$. Then, $\forall\beta\in\Supp(\algrA)$,
\begin{equation*}
[e_\alpha,e_\beta]_\eps=\sigma(\alpha,\beta)(1-(\eps\eps_\sigma^{-1})(\alpha,\beta))e_{\alpha+\beta}.
\end{equation*}
$\forall\beta$, $[e_\alpha,e_\beta]_\eps=0$ $\Leftrightarrow$ $\forall\beta$, $\eps(\alpha,\beta)=\eps_\sigma(\alpha,\beta)$, that is $\alpha\in\Gamma_{\eps,\eps_\sigma}$.
\end{proof}

Let us also define the set (potentially empty):
\begin{equation*}
R_{\eps_1,\eps_2}=\{i\in\Gamma,\,\forall j\in\Gamma,\, \eps_1(i-j,j)=\eps_2(i-j,j)\},
\end{equation*}
for two commutation factors $\eps_1$ and $\eps_2$ on $\Gamma$ over $\gC$.

\begin{proposition}
\label{prop-rset}
This set satisfies:
\begin{align*}
\forall i\in R_{\eps_1,\eps_2},&\quad -i\in R_{\eps_1,\eps_2}.\\
\forall i,j\in R_{\eps_1,\eps_2},&\quad i+j\in\Gamma_{\eps_1,\eps_2}.
\end{align*}
Moreover, if $\psi_{\eps_1}=\psi_{\eps_2}$, then $R_{\eps_1,\eps_2}=\Gamma_{\eps_1,\eps_2}$, else $R_{\eps_1,\eps_2}\cap \Gamma_{\eps_1,\eps_2}=\emptyset$, where $\psi_\eps$ is the signature function of $\eps$, defined in subsection \ref{subsec-commfact}.
\end{proposition}

\begin{proposition}[$\eps$-traces]
\label{prop-tracefine}
The $\eps$-traces on $\algrA$ are the linear maps $\algA\to\gC$ vanishing outside of
\begin{equation*}
\bigoplus_{\alpha\in R_{\eps,\eps_\sigma}}\algA^\alpha.
\end{equation*}
\end{proposition}
\begin{proof}
Let $T:\algrA\to\gC$ be an $\eps$-trace. $\forall\alpha,\beta\in\Supp(\algrA)$, we have $T([e_\alpha,e_\beta]_\eps)=0$. Then, from the proof of Proposition \ref{prop-centfine}, one obtains $(1-\eps\eps_\sigma^{-1}(\alpha,\beta))T(e_{\alpha+\beta})=0$. After a change of variables $\alpha\to\alpha-\beta$,
\begin{equation*}
\forall\alpha,\beta\in\Supp(\algrA),\quad (1-\eps\eps_\sigma^{-1}(\alpha-\beta,\beta))T(e_\alpha)=0.
\end{equation*}
And then, $\forall\alpha\in\Supp(\algrA)$, ($T(e_\alpha)=0$ or $\alpha\in R_{\eps,\eps_\sigma}$).
\end{proof}

Let us define the coordinates $(x_\alpha)_{\alpha\in\Supp(\algrA)}$ of a homogeneous $\eps$-derivation $\kX$ of $\algrA$ by $\forall\alpha\in\Supp(\algrA)$,
\begin{equation*}
\kX(e_\alpha)=\sigma(|\kX|,\alpha)\ x_\alpha\ e_{\alpha+|\kX|}.
\end{equation*}

\begin{theorem}[$\eps$-derivations]
The coordinates of a homogeneous $\eps$-derivation $\kX$ of $\algrA$ satisfy: $\forall\alpha,\beta\in\Supp(\algrA)$,
\begin{equation*}
x_{\alpha+\beta}=x_\alpha+(\eps\eps_\sigma^{-1})(|\kX|,\alpha)x_\beta.
\end{equation*}
Moreover, $\kX$ is inner $\Leftrightarrow$ $x_\alpha$ is proportional (independently of $\alpha$) to $1-(\eps\eps_\sigma^{-1})(|\kX|,\alpha)$.

The exact sequence of $\eps$-Lie algebras and $\caZ_\eps^\bullet(\algA)$-modules is canonically split:
\begin{gather*}
\xymatrix@1@C=25pt{{\algzero} \ar[r] & {\Int^\bullet_\eps(\algA)} \ar[r] & {\Der^\bullet_\eps(\algA)} \ar[r] & {\Out^\bullet_\eps(\algA)} \ar[r] & {\algzero}}
\end{gather*}
 and there are only two possibilities for $\kX$:
\begin{enumerate}
\item If $|\kX|\in\Gamma_{\eps,\eps_\sigma}$, $\kX$ is outer and given by a group morphism of $\Supp(\algrA)$ into $\gC$.
\item Otherwise, $\kX$ is inner.
\end{enumerate}
\end{theorem}
\begin{proof}
\begin{itemize}
\item Let $\kX$ be an homogeneous $\eps$-derivation of $\algrA$. The defining relation \eqref{defderiv} of $\eps$-derivations: $\forall\alpha,\beta\in\Supp(\algrA)$,
\begin{equation*}
\kX(e_\alpha\fois e_\beta)=\kX(e_\alpha)\fois e_\beta+\eps(|\kX|,\alpha)e_\alpha\fois\kX(e_\beta),
\end{equation*}
can be reexpressed in terms of the coordinates of $\kX$:
\begin{equation*}
\sigma(|\kX|,\alpha+\beta)x_{\alpha+\beta}\sigma(\alpha,\beta)=\sigma(|\kX|,\alpha)x_\alpha\sigma(\alpha+|\kX|,\beta)+\eps(|\kX|,\alpha)\sigma(|\kX|,\beta)x_\beta\sigma(\alpha,\beta+|\kX|).
\end{equation*}
Thanks to \eqref{eq-eps-deffactorset} and \eqref{eq-multcomm}, one obtains
\begin{equation*}
x_{\alpha+\beta}=x_\alpha+(\eps\eps_\sigma^{-1})(|\kX|,\alpha)x_\beta.
\end{equation*}
\item If $\kX$ is inner, then there exists $\lambda\in\gC$ such that $\forall\alpha\in\Supp(\algrA)$,
\begin{equation*}
\kX(e_\alpha)=\lambda[e_{|\kX|},e_\alpha]_\eps=\lambda\sigma(|\kX|,\alpha)(1-\eps\eps_\sigma^{-1}(|\kX|,\alpha))e_{|\kX|+\alpha}.
\end{equation*}
Due to the definition of the coordinates $(x_\alpha)$ of $\kX$, one has
\begin{equation*}
x_\alpha=\lambda(1-\eps\eps_\sigma^{-1}(|\kX|,\alpha)).
\end{equation*}
The converse is straightforward.
\item If $|\kX|\in\Gamma_{\eps,\eps_\sigma}$, $x_{\alpha+\beta}=x_\alpha+x_\beta$ and $x:\Supp(\algrA)\to\gC$ is a group morphism. Moreover, the only possible inner derivation is $0$. Consequently, $\kX$ can be seen as an outer $\eps$-derivation.
\item If $|\kX|\notin\Gamma_{\eps,\eps_\sigma}$, there exists $\alpha\in\Supp(\algrA)$ such that $\eps\eps_\sigma^{-1}(|\kX|,\alpha)\neq1$. Since $\forall\alpha,\beta\in\Supp(\algrA)$,
\begin{equation*}
x_\alpha+\eps\eps_\sigma^{-1}(|\kX|,\alpha)x_\beta=x_{\alpha+\beta}=x_\beta+\eps\eps_\sigma^{-1}(|\kX|,\beta)x_\alpha,
\end{equation*}
one concludes that $\forall\beta\in\Supp(\algrA)$,
\begin{equation*}
x_\beta=\frac{1-\eps\eps_\sigma^{-1}(|\kX|,\beta)}{1-\eps\eps_\sigma^{-1}(|\kX|,\alpha)}x_\alpha,
\end{equation*}
and, by the above property of inner $\eps$-derivations, $\kX$ is inner. The grading, and more precisely the belonging of the grading to $\Gamma_{\eps,\eps_\sigma}$ or not, provides a splitting of the $\eps$-derivations exact sequence.
\end{itemize}
\end{proof}

\begin{corollary}
We define $\ksl_\eps^\bullet(\algA)$ to be the space of traceless matrices (see Proposition \ref{prop-tracefine}) and the $\caZ_\eps^\bullet(\algA)$-module
\begin{equation*}
\mathcal{I}_\eps^\bullet(\algA)=\bigoplus_{\alpha\notin\Gamma_{\eps,\eps_\sigma}}\algA^\alpha.
\end{equation*}
Then,
\begin{itemize}
 \item $\algrA=\caZ_\eps^\bullet(\algA)\oplus \caI_\eps^\bullet(\algA)$
\item $\ad: \caI_\eps^\bullet(\algA)\to\Int_\eps^\bullet(\algA)$
is an isomorphism of $\caZ_\eps^\bullet(\algA)$-module. \item
Moreover, if $\eps$ is a proper commutation factor,
$\caI_\eps^\bullet(\algA)=\ksl_\eps^\bullet(\algA)$, and $\ad$
restricted to this $\eps$-Lie algebra is an isomorphism of
$\eps$-Lie algebra.
\end{itemize}
\end{corollary}
\begin{proof}
Indeed, $\text{Ker}(\ad)=\caZ_\eps^\bullet(\algA)$. If $\eps$ is proper, $R_{\eps,\eps_\sigma}=\Gamma_{\eps,\eps_\sigma}$ due to Proposition \ref{prop-rset}. Then, Proposition \ref{prop-tracefine} gives the result: $\mathcal{I}_\eps^\bullet(\algA)=\ksl_\eps^\bullet(\algA)$.
\end{proof}

\begin{corollary}
\label{cor-fineresults}
In the case $\eps=\eps_\sigma$, one has
\begin{align*}
\caZ_{\eps_\sigma}^\bullet(\algA)&=\algrA.\\
\Der_{\eps_\sigma}^\bullet(\algA)&=\Out_{\eps_\sigma}^\bullet(\algA).
\end{align*}
Moreover, the $\eps_\sigma$-traces of $\algrA$ are the linear maps $\algA\to\gC$.
\end{corollary}
\begin{proof}
Indeed, $\Gamma_{\eps_\sigma,\eps_\sigma}=R_{\eps_\sigma,\eps_\sigma}=\Gamma$. Note that $\forall\alpha,\beta\in\Supp(\algrA)$, $[e_\alpha,e_\beta]_{\eps_\sigma}=0$.
\end{proof}

We notice that the results of Corollary \ref{cor-fineresults} are very different from the one obtained in subsection \ref{subsec-exelem}, namely Proposition \ref{prop-centermatrix} and Proposition \ref{prop-derivmatrix}, which are close to the results of the non-graded case on the matrices. This shows that the extension to the $\eps$-graded case of the framework of derivations is non-trivial and can provide interesting examples.

Note that the dimension has not been used for the moment in this subsection, so that the latter results remain true for any fine-graded division algebra. If we take into account that $\algrA$ is finite dimensional, one obtains:
\begin{corollary}
For $\algrA=\Matr_D$ with fine grading, and $D<\infty$,
\begin{equation*}
\Out_\eps^\bullet(\algA)=\algzero.
\end{equation*}
\end{corollary}
\begin{proof}
If $\Supp(\algrA)$ is finite, the only group morphism $\Supp(\algrA)\to\gC$ is trivial.
\end{proof}

\subsubsection{\texorpdfstring{Differential calculus and $\eps$-connections}{Differential calculus and epsilon-connections}}
\label{subsub-diffcalcfine}

To describe the differential calculus of $\algrA=\Matr_D$ based on
$\eps$-derivations, we need to introduce the following notation:
\begin{equation*}
\Omega_{\eps,\caZ}^{\bullet,\bullet}(\caI_\eps(\algA),\algA)=\algrA\otimes_{\caZ_\eps(\algA)}\exter^\bullet_{\eps,\caZ_\eps(\algA)}(\caI_\eps(\algA)^\ast)^\bullet
\end{equation*}
inspired of the part \ref{subsub-diffcalcelem}, and where the
$\eps$-exterior algebra of $\caI^\bullet_\eps(\algA)$ is made of
tensor products on $\caZ_\eps^\bullet(\algA)$. As for Proposition
\ref{prop-gv}, it is easy to show that
$\Omega_{\eps,\caZ}^{\bullet,\bullet}(\caI_\eps(\algA),\algA)$ is
a $(\gZ\times\Gamma)$-graded algebra, and a graded differential
algebra if $\caI_\eps^\bullet(\algA)$ is an $\eps$-Lie algebra.

\begin{lemma}
One can define the coordinates of any element
$\omega\in\Omr^{n,|\omega|}(\algA)$ of the differential calculus:
\begin{multline*}
\omega(\ad_{e_{\alpha_1}},\dots,\ad_{e_{\alpha_n}})=\omega_{\alpha_1,\dots,\alpha_n}e_{|\omega|}\fois
e_{\alpha_1}\fois\dots\fois e_{\alpha_n}\\
=\sigma(|\omega|,\alpha_1+\dots+\alpha_n)\sigma(\alpha_1,\alpha_2+\dots+\alpha_n)\dots\sigma(\alpha_{n-1},\alpha_n)\,
\omega_{\alpha_1,\dots,\alpha_n}e_{|\omega|+\alpha_1+\dots+\alpha_n},
\end{multline*}
for $\alpha_i\in\Supp(\algrA)\backslash\Gamma_{\eps,\eps_\sigma}$.
Then, these coordinates satisfy:
\begin{align*}
\omega_{\alpha_1,\dots,\alpha_i,\alpha_{i+1},\dots,\alpha_n}
&=-(\eps\eps_\sigma^{-1})(\alpha_i,\alpha_{i+1})\,\omega_{\alpha_1,\dots,\alpha_{i+1},\alpha_i,\dots,\alpha_n}\\
\omega_{\alpha_1,\dots,\alpha_i+\beta,\dots,\alpha_n}
&=\omega_{\alpha_1,\dots,\alpha_i,\dots,\alpha_n}\quad\forall\beta\in\Gamma_{\eps,\eps_\sigma}\cap
\Supp(\algrA).
\end{align*}
\end{lemma}
\begin{proof}
The proof is straightforward by using the axioms of the
differential calculus in equation \eqref{defcalcdiff}.
\end{proof}

\begin{theorem}
Let $\algrA=\Matr_D$ be an $\eps$-graded matrix algebra with fine
grading ($D<\infty$). Then:
\begin{itemize}
\item
$\ad_\ast:\Omega_{\eps,\caZ}^{\bullet,\bullet}(\caI_\eps(\algA),\algA)\to\Omr^{\bullet,\bullet}(\algA)$
is an isomorphism of $(\gZ\times\Gamma)$-graded algebras. \item
The space of $\eps$-connections on $\algrA$ is an affine space
modeled on the vector space
\begin{equation*}
\Omr^{1,0}(\algA)\approx\bigoplus_{\alpha\notin\Gamma_{\eps,\eps_\sigma}}\algA^\alpha\otimes_{\caZ_\eps(\algA)}(\algA^\alpha)^\ast.
\end{equation*}
\item Moreover, if $\eps$ is proper, the above map $\ad_\ast$ is
an isomorphism of graded differential algebras.
\end{itemize}
\end{theorem}
\begin{proof}
Since all $\eps$-derivations are inner, one can adapt the proof of
Theorem \ref{thm-calcdiff}.
\end{proof}

\subsubsection{Concrete examples}
\label{subsub-concretfine}

Let us apply the latter results on simple examples of fine-graded
matrix algebras.
\begin{example}
\begin{itemize}
\item For $\Gamma=\algzero$, the only possibility of fine-graded complex matrix algebra is $\algrA=\gC$.
\item For $\Gamma=\gZ_2$, there is no fine-graded matrix algebra $\algrA$ such that $\Supp(\algrA)=\gZ_2$.
\end{itemize}
\end{example}

\begin{example}
For $\Gamma=\gZ_2\times\gZ_2$, the only possibility of fine-graded complex matrix algebra $\algrA$ such that $\Supp(\algrA)=\Gamma$, is $\algrA=\Matr_2$, the two by two complex matrix algebra.
\begin{itemize}
\item $\algA^{(0,0)}=\gC\gone$, and, up to a permutation, $\algA^{(1,0)}=\gC\tau_1$, $\algA^{(0,1)}=\gC\tau_2$ and $\algA^{(1,1)}=\gC\tau_3$, where the $\tau_i$ are the Pauli matrices:
\begin{equation*}
\tau_1=\begin{pmatrix} 0 & 1  \\ 1 & 0 \end{pmatrix},\qquad \tau_2=\begin{pmatrix} 0 & -i  \\ i & 0 \end{pmatrix},\qquad \tau_3=\begin{pmatrix} 1 & 0  \\ 0 & -1 \end{pmatrix}.
\end{equation*}
\item The factor set $\sigma$ associated to the algebra $\algrA$ is given by:
\begin{multline*}
\sigma((1,0),(0,1))=-\sigma((0,1),(1,0))=-\sigma((1,0),(1,1))\\
=\sigma((1,1),(1,0)) =\sigma((0,1),(1,1))=-\sigma((1,1),(0,1))=i,
\end{multline*}
for the non-trivial terms. The associated proper commutation factor is then: $\forall j,k\in\Gamma$,
\begin{equation*}
\eps_\sigma(j,k)=(-1)^{j_1k_2+j_2k_1}.
\end{equation*}
\item Since $\Gamma$ is a product of cyclic groups, and $\gC$ is a field of characteristic zero, there is no non-zero group morphism $\Gamma\to\gC$.
\item If $\algrA$ is endowed with its associated commutation factor $\eps_\sigma$, it is a color algebra. Then, by Corollary \ref{cor-fineresults}, one obtains:
\begin{equation*}
\caZ^\bullet_{\eps_\sigma}(\Matr_2)=\Matr_2,\qquad \Der^\bullet_{\eps_\sigma}(\Matr_2)=\Out^\bullet_{\eps_\sigma}(\Matr_2)=\algzero,
\end{equation*}
and the $\eps_\sigma$-traces are the linear maps $\Matr_2\to\gC$. Moreover, the differential calculus is trivial:
\begin{equation*}
\underline\Omega^{\bullet,\bullet}_{\eps_\sigma}(\Matr_2)=\underline\Omega^{0,\bullet}_{\eps_\sigma}(\Matr_2)=\Matr_2.
\end{equation*}
\item If $\algrA$ is endowed with the commutation factor $\eps(j,k)=(-1)^{j_1k_1+j_2k_2}$, it is a color superalgebra. Since $\Gamma_{\eps,\eps_\sigma}=\{(0,0),(1,1)\}$ and $R_{\eps,\eps_\sigma}=\{(1,0),(0,1)\}$, one has:
\begin{equation*}
\caZ_{\eps}^\bullet(\Matr_2)=\gC\gone\oplus\gC\tau_3,\qquad \Der_{\eps}^\bullet(\Matr_2)=\Int_\eps^\bullet(\Matr_2) =\gC\,\ad_{\tau_1}\oplus\gC\,\ad_{\tau_2},
\end{equation*}
and the $\eps$-traces are the linear maps $\caZ_{\eps}(\Matr_2)\to\gC$. One can also notice that $\gC\tau_1\oplus\gC\tau_2=\ksl^1(1,1)$ in the notations of Example \ref{ex-supermatrix}, so that the differential calculus writes:
\begin{equation*}
\underline\Omega^{\bullet,\bullet}_{\eps}(\Matr_2)=\Matr_2\otimes_{\ksl^0(1,1)}\Big(\symes^\bullet\ksl^1(1,1)^\ast\Big).
\end{equation*}
\end{itemize}
\end{example}

Note that it has been proved in \cite{Bahturin:2002} that every graded matrix algebra can be decomposed as the tensor product of an elementary graded matrix algebra and a fine graded matrix algebra, which are the two typical cases of graded matrix algebras (studied in subsection\ref{subsec-exelem} and \ref{subsec-exfine}). In this subsection, we have seen that fine graded matrix algebras are naturally related to commutation factors and that the theory of comparison of commutation factors permits to characterize properties of such algebras.

\subsection{Application to the Moyal algebra}
\label{subsec-moyal}

The Moyal algebra $\caM$ is a deformation quantization of $\gR^D$. Expanded on the matrix basis, the Moyal algebra can be seen as a generalization of finite-dimensional matrix algebras. Without considering all the possible gradings and commutation factors, we will construct a superalgebra $\algrA_\theta$ from $\caM$, study its properties, and apply it to explain mathematically a noncommutative gauge theory used in mathematical physics \cite{deGoursac:2007gq}.

\subsubsection{General properties of the Moyal algebra}

Here, we collect the main properties of the Moyal algebra which will be used in the subsequent discussion. For more details, see e.g \cite{GraciaBondia:1987kw,Varilly:1988jk}. Let $\caS\equiv\caS(\gR^D)$ with even dimension $D$ and $\caS^\prime$ be respectively the space of complex-valued Schwartz functions on $\gR^D$ and its topological dual, the space of tempered distributions on $\gR^D$. Consider a symplectic structure on $\gR^D$, represented by an invertible skew-symmetric $D\times D$ matrix $\Sigma$, and $\Theta=\theta \Sigma$, where the positive parameter $\theta$ has mass dimension $-2$. The Moyal product associated to $\Theta$ can be conveniently defined on $\caS\times\caS$ by: $\forall a,b\in\caS$,
\begin{equation}
(a\star b)(x)=\frac{1}{(\pi\theta)^D|\text{det}\Sigma|}\int d^Dyd^Dz\ a(x+y)b(x+z)e^{-2iy\Theta^{-1}z} \label{eq-moyal}
\end{equation}
such that $(a\star b)\in\caS$, where\footnote{We use in this subsection the implicit summation convention.} $\Theta.k\equiv\Theta_{\mu\nu}k_\nu$. The $\star$ product \eqref{eq-moyal} can be further extended to $\caS^\prime\times\caS$ upon using duality of linear spaces: $\langle T\star a,b \rangle = \langle T,a\star b\rangle$, $\forall T\in\caS^\prime$, $\forall a,b\in\caS$. In a similar way, \eqref{eq-moyal} can be extended to $\caS \times \caS^\prime$. Owing to the smoothening properties of \eqref{eq-moyal} together with the tracial identity $\int d^Dx\ (a\star b)(x)=\int d^Dx\ a(x).b(x)$, where the symbol ``.'' denotes the (commutative) usual pointwise product, it can be shown that $T\star a$ and $a \star T$ are smooth functions \cite{GraciaBondia:1987kw,Varilly:1988jk}. The Moyal algebra $\caM$ is then defined as
\begin{equation*}
\caM=\mathcal L\cap \mathcal R,
\end{equation*}
where $\mathcal L$ (resp. $\mathcal R$) is the subspace of $\caS^\prime$ whose multiplication from right (resp.\ left) by any Schwartz function is a subspace of $\caS$.

By denoting $\dag$ the complex conjugation, $(\caM,\star,{}^\dag)$ is a unital involutive topological algebra which involves in particular the ``coordinate'' functions $x_\mu$ satisfying the relation defined on $\caM$:
\begin{equation*}
[x_\mu,x_\nu]_\star =i\Theta_{\mu\nu},
\end{equation*}
where we set $[a,b]_\star=a\star b-b\star a$ and $\{a,b\}_\star=a\star b+b\star a$. Defining $\wx_\mu=2\Theta^{-1}_{\mu\nu}x_\nu$, other relevant properties of the $\star$-product that hold on $\caM$ and that will be used in the following are: $\forall a,b\in\caM$,
\begin{subequations}
\label{eq-relat-moyal}
\begin{align}
\partial_\mu(a\star b)&=\partial_\mu a\star b+a\star\partial_\mu b, \\
[\wx_\mu,a]_\star&=2i\partial_\mu a,\\
\{\wx_\mu,a\}_\star&=2\wx_\mu a,\\
[\wx_\mu\wx_\nu,a]_\star&=\frac i2\wx_\mu\partial_\nu a+\frac i2\wx_\nu\partial_\mu a.
\end{align}
\end{subequations}

\medskip

By defining the straightforward generalization of the $\phi^4$ theory on the euclidean Moyal space:
\begin{equation*}
S(\phi)=\int \dd^Dx\Big(\frac 12(\partial_\mu\phi)^2 +\frac{m^2}{2}\phi^2 +\lambda\phi\star\phi\star\phi\star\phi\Big),
\end{equation*}
one finds a new divergence, called Ultraviolet-Infrared (UV/IR) mixing \cite{Minwalla:1999px}, which spoils the renormalizability of the theory. A solution to this problem has been carried out by Grosse and Wulkenhaar by adding a harmonic term in the action:
\begin{equation}
S=\int \dd^Dx\Big(\frac 12(\partial_\mu\phi)^2+\frac{\Omega^2}{2}(\wx_\mu\phi)^2+\frac{m^2}{2}\phi^2+\lambda \phi^{\star 4}\Big),\label{actscalmoy}
\end{equation}
so that the theory becomes renormalizable to all orders in perturbation \cite{Grosse:2004yu}. A new gauge theory on the Moyal space has been exhibited in \cite{deGoursac:2007gq,Grosse:2007dm} from a one-loop effective action of the model \eqref{actscalmoy}:
\begin{equation}
S=\frac 14\int d^Dx\Big(F_{\mu\nu}\star F_{\mu\nu}+\Omega'^2\{\caA_\mu,\caA_\nu\}_\star^2+\kappa\caA_\mu\star\caA_\mu\Big),\label{actgaugemoy}
\end{equation}
where $F_{\mu\nu}=\partial_\mu A_\nu-\partial_\nu A_\mu-i[A_\mu,A_\nu]_\star$ and $\caA_\mu=A_\mu+\frac 12\wx_\mu$. This gauge action is indeed a good candidate for a renormalizable gauge theory on the Moyal space.

\subsubsection{\texorpdfstring{An $\eps$-graded algebra constructed from Moyal space}{An epsilon-graded algebra constructed from Moyal space}}
\label{subsub-moyalalg}

We construct here an $\eps$-graded algebra $\algrA_\theta$ (here a superalgebra) from the Moyal algebra $\caM$. We also study its $\eps$-center and exhibit a particular $\eps$-Lie subalgebra $\kg^\bullet$ of its inner derivations. From these considerations, we find as a classical scalar action the quadratic part of the renormalizable Grosse-Wulkenhaar model \eqref{actscalmoy}.

\begin{definition}
Let $\algrA_\theta$ be the $\gZ_2$-graded complex vector space defined by $\algA^0_\theta=\caM$ and $\algA^1_\theta=\caM$. Let us also introduce the following product on $\algrA_\theta=\algA^0_\theta\oplus\algA^1_\theta$: $\forall a,b,c,d\in\caM$,
\begin{equation}
(a,b)\fois (c,d)=(a\star c+\alpha\ b\star d,a\star d+b\star c),\label{defprodmoy}
\end{equation}
where $\alpha$ is a real parameter. Finally, for $i,j\in\gZ_2$, define the usual commutation factor: $\eps(i,j)=(-1)^{ij}$.
\end{definition}
\begin{proposition}
\label{prop-moy1}
The vector space $\algrA_\theta$, endowed with the product \eqref{defprodmoy} and the commutation factor $\eps$, is an $\eps$-graded algebra. The bracket of its associated $\eps$-Lie algebra (which is a superalgebra) has the following expression: for $\phi=(\phi_0,\phi_1)\in\algrA_\theta$ and $\psi=(\psi_0,\psi_1)\in\algrA_\theta$,
\begin{equation*}
[\phi,\psi]_\eps=([\phi_0,\psi_0]_\star+\alpha\{\phi_1,\psi_1\}_\star,[\phi_0,\psi_1]_\star+[\phi_1,\psi_0]_\star).
\end{equation*}
Moreover, $(\phi_0,\phi_1)^\ast=(\phi_0^\dag,\phi_1^\dag)$ is an involution for $\algrA_\theta$ ($\alpha\in\gR$).
\end{proposition}
\begin{proof}
For example, let us check the associativity of the product: $\forall a,b,c,d,e,f\in\caM$,
\begin{align*}
&((a,b)\fois(c,d))\fois(e,f)\\
&=(a\star c\star e+\alpha(b\star d\star e+a\star d\star f+b\star c\star f), a\star c\star f+a\star d\star e+b\star c\star e+\alpha\ b\star d\star f)\\
&=(a,b)\fois((c,d)\fois(e,f)).
\end{align*}
In the same way, distributivity and the other axioms are verified. Notice that $\gone=(1,0)$ is the unit element of this algebra.
\end{proof}

Since the center of the Moyal algebra is trivial ($\caZ(\caM)=\gC$), we find the $\eps$-center of $\algrA_\theta$:
\begin{equation*}
\caZ_\eps^\bullet(\algA_\theta)=\gC\gone=\gC\oplus\algzero.
\end{equation*}
Furthermore, a (non-graded) trace of a subclass of this algebra is given by: $\forall\phi=(\phi_0,\phi_1)\in\algrA_\theta$,
\begin{equation*}
\tr(\phi)=\int d^Dx\ \phi_0(x).
\end{equation*}

Let us now define the following objects $\gamma=1$, $\xi_\mu=-\frac 12\wx_\mu$ and $\eta_{\mu\nu}=\frac12\wx_\mu\wx_\nu=2\xi_\mu\xi_\nu$, and give some calculation rules, deduced from \eqref{eq-relat-moyal}: $\forall\phi\in\caM$,
\begin{align}
[i\gamma,\phi]_\star=0,\qquad  \{i\gamma,\phi\}_\star=2i\gamma\phi,\nonumber\\
[i\xi_\mu,\phi]_\star=\partial_\mu\phi, \qquad \{i\xi_\mu,\phi\}_\star=2i\xi_\mu.\phi,\qquad [i\eta_{\mu\nu},\phi]_\star=\frac 12\xi_\mu\partial_\nu\phi+\frac 12\xi_\nu\partial_\mu\phi.\label{calcrules}
\end{align}
Consequently, the usual derivations of the Moyal algebra $\partial_\mu$ are inner and can be expressed in terms of $\xi_\mu$.
\begin{proposition}
\label{prop-moy2}
Then, $\ad_{(0,i\gamma)}$, $\ad_{(i\xi_\mu,0)}$, $\ad_{(0,i\xi_\mu)}$ and $\ad_{(i\eta_{\mu\nu},0)}$ are real $\eps$-derivations of $\algrA_\theta$ of respective degrees 1, 0, 1 and 0. Moreover, the vector space $\kg^\bullet$ generated by these $\eps$-derivations is an $\eps$-Lie subalgebra of $\Der_\eps^\bullet(\algA_\theta)$ and a right $\caZ_\eps^\bullet(\algA_\theta)$-module.
\end{proposition}
\begin{proof}
The following relations, computed from \eqref{calcrules},
\begin{subequations}
\begin{align*}
[(0,i\gamma),(0,i\gamma)]_\eps&=(-2\alpha,0)\\
[(i\xi_\mu,0),(0,i\gamma)]_\eps&=(0,0)\\
[(0,i\xi_\mu),(0,i\gamma)]_\eps&=(-2\alpha\xi_\mu,0)\\
[(i\eta_{\mu\nu}),(0,i\gamma)]_\eps&=(0,0)\\
[(i\xi_\mu,0),(i\xi_\nu,0)]_\eps&=(i\Theta^{-1}_{\mu\nu},0)\\
[(i\xi_\mu,0),(0,i\xi_\nu)]_\eps&=(0,i\Theta^{-1}_{\mu\nu}\gamma)\\
[(0,i\xi_\mu),(0,i\xi_\nu)]_\eps&=(-\alpha\eta_{\mu\nu},0)\\
[(i\eta_{\mu\nu},0),(i\xi_\rho,0)]_\eps&=(\frac i2\xi_\mu\Theta^{-1}_{\nu\rho}+\frac i2\xi_\nu\Theta^{-1}_{\mu\rho},0)\\
[(i\eta_{\mu\nu},0),(0,i\xi_\rho)]_\eps&=(0,\frac i2\xi_\mu\Theta^{-1}_{\nu\rho}+\frac i2\xi_\nu\Theta^{-1}_{\mu\rho})\\
[(i\eta_{\mu\nu},0),(i\eta_{\rho\sigma},0)]_\eps&=(\frac i2\eta_{\mu\rho}\Theta^{-1}_{\nu\sigma} +\frac i2\eta_{\mu\sigma}\Theta^{-1}_{\nu\rho} +\frac i2\eta_{\nu\rho}\Theta^{-1}_{\mu\sigma} +\frac i2\eta_{\nu\sigma}\Theta^{-1}_{\mu\rho},0),
\end{align*}
\end{subequations}
combined with $\forall a,b\in\algrA_\theta$, $[\ad_a,\ad_b]_\eps=\ad_{[a,b]_\eps}$, show that $\kg^\bullet$ is an $\eps$-Lie algebra.
\end{proof}
Notice that the vector space generated only by $\ad_{(i\xi_\mu,0)}$ and $\ad_{(0,i\xi_\mu)}$ is not an $\eps$-Lie subalgebra. Indeed, $\kg^\bullet$ is the smallest subalgebra of $\Der_\eps^\bullet(\algA_\theta)$ involving $\ad_{(i\xi_\mu,0)}$ and $\ad_{(0,i\xi_\mu)}$.

\medskip
We now make use of the previous construction of the $\eps$-graded algebra $\algrA_\theta$ to give a mathematical interpretation to the theory \eqref{actscalmoy}. In the notations of Propositions \ref{prop-moy1} and \ref{prop-moy2}, one can compute: $\forall\phi\in\caM$ real,
\begin{multline*}
\tr\Big(|[(i\xi_\mu,0),(\phi,\phi)]_\eps|^2  +|[(0,i\xi_\mu),(\phi,\phi)]_\eps|^2 +\frac{1}{\theta}|[(0,i\gamma),(\phi,\phi)]_\eps|^2\Big)\\
=\int d^Dx\Big((1+2\alpha)(\partial_\mu\phi)^2+\alpha^2(\wx_\mu\phi)^2+\frac{4\alpha^2}{\theta}\phi^2\Big),
\end{multline*}
where $|a|^2=a^\ast\fois a\in\algrA_\theta$ and we impose $\phi_0=\phi_1=\phi$. Setting $\Omega^2=\frac{\alpha^2}{1+2\alpha}$ and $m^2=\frac{4\alpha^2}{\theta(1+2\alpha)}$, we find that the latter expression is the quadratic part of the Grosse-Wulkenhaar action with harmonic term \eqref{actscalmoy}, which therefore seems to be related to $\Omr^{\bullet,\bullet}(\algA_\theta|\kg)$.

\begin{remark}
This stems to the fact that the graduation of the ``graded Moyal algebra'' $\algrA_\theta$ mimics the Langmann-Szabo duality \cite{Langmann:2002cc} which seems to play an important role in the renormalizability of this model. Indeed, the Langmann-Szabo duality ($\partial_\mu\rightleftarrows\wx_\mu$) can be related to ($i[\xi_\mu,.]_\star\rightleftarrows\{\xi_\mu,.\}_\star$) or to ($(i\xi_\mu,0)\rightleftarrows(0,i\xi_\mu)$) thanks to the following relations
\begin{align*}
[(i\xi_\mu,0),(\phi_0,\phi_1)]_\eps&=([i\xi_\mu,\phi_0]_\star,[i\xi_\mu,\phi_1]_\star),\\
[(0,i\xi_\mu),(\phi_0,\phi_1)]_\eps&=(i\alpha\{\xi_\mu,\phi_1\}_\star,[i\xi_\mu,\phi_0]_\star).
\end{align*}
and further assuming $\phi_0=\phi_1=\phi$. In fact, an exact correspondence between Langmann-Szabo duality and this grading exchange has been proven in \cite{deGoursac:2010zb}. Moreover, the previous formalism, contrary to the usual ($\partial_\mu\rightleftarrows\wx_\mu$) symmetry, can also be applied to gauge theory, as we will see in the next part.
\end{remark}

\subsubsection{\texorpdfstring{$\eps$-connections for this $\eps$-graded algebra}{epsilon-connections for this epsilon-graded algebra}}
\label{subsub-moyalconn}

Let us consider the differential calculus $\Omr^{\bullet,\bullet}(\algA_\theta|\kg)$ and apply the considerations of subsection~\ref{subsec-epsconn} for the module $\modrM=\algrA_\theta$, with the hermitian structure $\langle a,b\rangle=a^\ast\fois b$, $\forall a,b\in\algrA_\theta$. We describe here the $\eps$-connections and their curvatures in terms of the gauge potentials and the covariant coordinates we introduce. We also look at the action of gauge transformations on such objects. Finally, we recover the recently constructed candidate \eqref{actgaugemoy} for a renormalizable noncommutative gauge theory as an action built from the curvature presented in this subsection.

\begin{proposition}
Let $\nabla$ be a hermitian $\eps$-connection on $\algrA_\theta$. For $\kX\in\Der_\eps^\bullet(\algA_\theta)$, define the gauge potential associated to $\nabla$: $-iA_\kX=\nabla(\gone)(\kX)$. $A:\kX\mapsto A_\kX$ is in $\Omr^{1,0}(\algA_\theta|\kg)$ and if $\kX$ is real, $A_\kX$ is real too.

Then, the $\eps$-connection $\nabla$ takes the form: $\forall\kX\in\Der_\eps^\bullet(\algA_\theta)$, $\forall a\in\algrA_\theta$,
\begin{equation*}
\nabla_\kX a=\kX(a)-iA_\kX\fois a,
\end{equation*}
where $\nabla_\kX a=\eps(|\kX|,|a|)\nabla(a)(\kX)$. Defining $F_{\kX,\kY}\fois a=i\eps(|\kX|+|\kY|,|a|)R(a)(\kX,\kY)$, we obtain the curvature: $\forall\kX,\kY\in\Der_\eps^\bullet(\algA_\theta)$,
\begin{equation*}
F_{\kX,\kY}=\kX(A_\kY)-\eps(|\kX|,|\kY|)\kY(A_\kX)-i[A_\kX,A_\kY]_\eps-A_{[\kX,\kY]_\eps}.
\end{equation*}
\end{proposition}
\begin{proof}
Indeed, for $\kX\in\Der_\eps^\bullet(\algA_\theta)$ and $a\in\algrA_\theta$,
\begin{equation*}
\nabla(a)(\kX)=\nabla(\gone\fois a)(\kX)=\nabla(\gone)(\kX)+\gone\fois \dd a(\kX).
\end{equation*}
The curvature simplifies as:
\begin{align*}
R(a)(\kX,\kY)=&\eps(|\kX|,|\kY|)\nabla(\nabla(a)(\kY))(\kX)-\nabla(\nabla(a)(\kX))(\kY)-\nabla(a)([\kX,\kY]_\eps)\\
=&\eps(|\kX|,|\kY|)\dd(\dd a(\kY))(\kX)-\dd(\dd a(\kX))(\kY)-\dd a([\kX,\kY]_\eps)\\
&-i\eps(|\kX|+|a|,|\kY|)\dd(A_\kY\fois a)(\kX)-i\eps(|a|,|\kX|)A_\kX\fois\dd a(\kY)-\eps(|a|,|\kX|+|\kY|)A_\kX\fois A_\kY\fois a\\
&+i\eps(|a|,|\kX|)\dd(A_\kX\fois a)(\kY)+i\eps(|\kX|+|a|,|\kY|)A_\kY\fois\dd a(\kX)\\
&+\eps(|\kX|,|\kY|)\eps(|a|,|\kX|+|\kY|)A_\kY\fois A_\kX\fois a+i\eps(|a|,|\kX|+|\kY|)A_{[\kX,\kY]_\eps}\fois a.
\end{align*}
By using
\begin{equation}
0=\dd^2a(\kX,\kY)=\eps(|\kX|,|\kY|)\dd(\dd a(\kY))(\kX)-\dd(\dd a(\kX))(\kY)-\dd a([\kX,\kY]_\eps),\label{d2ident}
\end{equation}
we find:
\begin{equation*}
R(a)(\kX,\kY)=\eps(|a|,|\kX|+|\kY|)(-i\kX(A_\kY)+i\eps(|\kX|,|\kY|)\kY(A_\kX)-[A_\kX,A_\kY]_\eps+iA_{[\kX,\kY]_\eps})\fois a.
\end{equation*}
\end{proof}

Therefore, we set:
\begin{subequations}
\label{potmoy}
\begin{align}
\nabla(\gone)(\ad_{(i\xi_\mu,0)})&=(-iA_\mu^0,0),&
\nabla(\gone)(\ad_{(0,i\xi_\mu)})&=(0,-iA_\mu^1),\\
\nabla(\gone)(\ad_{(0,i\gamma)})&=(0,-i\varphi),&
\nabla(\gone)(\ad_{(i\eta_{\mu\nu},0)})&=(-iG_{\mu\nu},0).
\end{align}
\end{subequations}
The associated curvature can then be expressed as:
\begin{subequations}
\label{subeq-moy1}
\begin{align}
F_{(0,i\gamma),(0,i\gamma)}&=(2i\alpha\varphi-2i\alpha\varphi\star\varphi,0)\\
F_{(i\xi_\mu,0),(0,i\gamma)}&=(0,\partial_\mu\varphi-i[A_\mu^0,\varphi]_\star)\\
F_{(0,i\xi_\mu),(0,i\gamma)}&=(-i\alpha\wx_\mu\varphi-i\alpha\{A_\mu^1,\varphi\}_\star+2i\alpha(A_\mu^1-A_\mu^0),0)\\
F_{(i\eta_{\mu\nu},0),(0,i\gamma)}&=(0,-\frac 14\wx_\mu\partial_\nu\varphi-\frac 14\wx_\nu\partial_\mu\varphi-i[G_{\mu\nu},\varphi]_\star)\\
F_{(i\xi_\mu,0),(i\xi_\nu,0)}&=(\partial_\mu A_\nu^0-\partial_\nu A_\mu^0-i[A_\mu^0,A_\nu^0]_\star,0)\\
F_{(i\xi_\mu,0),(0,i\xi_\nu)}&=(0,\partial_\mu A_\nu^1-\partial_\nu A_\mu^0-i[A_\mu^0,A_\nu^1]_\star-\Theta_{\mu\nu}^{-1}\varphi)\\
F_{(0,i\xi_\mu),(0,i\xi_\nu)}&=(-i\alpha\wx_\mu A_\nu^1-i\alpha\wx_\nu A_\mu^1-i\alpha\{A_\mu^1,A_\nu^1\}_\star-i\alpha G_{\mu\nu},0)\\
F_{(i\xi_\mu,0),(i\eta_{\nu\rho},0)}&=(\partial_\mu G_{\nu\rho}+\wx_\nu\partial_\rho A_\mu^0+\wx_\rho\partial_\nu A_\mu^0-i[A_\mu^0,G_{\nu\rho}]_\star+2\Theta^{-1}_{\nu\mu}A_\rho^0+2\Theta^{-1}_{\rho\mu}A_\nu^0,0)\\
F_{(0,i\xi_\mu),(i\eta_{\nu\rho},0)}&=(0,\partial_\mu G_{\nu\rho}+\wx_\nu\partial_\rho A_\mu^1+\wx_\rho\partial_\nu A_\mu^1-i[A_\mu^1,G_{\nu\rho}]_\star+2\Theta^{-1}_{\nu\mu}A_\rho^1+2\Theta^{-1}_{\rho\mu}A_\nu^1)\\
F_{(i\eta_{\mu\nu},0),(i\eta_{\rho\sigma},0)}&=(-\wx_\mu\partial_\nu G_{\rho\sigma}-\wx_\nu\partial_\mu G_{\rho\sigma}+\wx_\rho\partial_\sigma G_{\mu\nu}+\wx_\sigma\partial_\rho G_{\mu\nu}-i[G_{\mu\nu},G_{\rho\sigma}]_\star\nonumber\\
&-2\Theta^{-1}_{\mu\rho}G_{\nu\sigma}-2\Theta^{-1}_{\nu\rho}G_{\mu\sigma}-2\Theta^{-1}_{\mu\sigma}G_{\nu\rho} -2\Theta^{-1}_{\nu\sigma}G_{\mu\rho},0),
\end{align}
\end{subequations}
where we have suppressed the $ad$'s in the indices of $F$ to simplify the notations.

\begin{proposition}
The unitary gauge transformations $\Phi$ of $\algrA_\theta$ are completely determined by $g=\Phi(\gone)$: $\forall a\in\algrA_\theta$, $\Phi(a)=g\fois a$, and $g$ is a unitary element of $\algrA_\theta$ of degree 0: $g^\ast\fois g=\gone$. Then, this gauge transformation acts on the gauge potential and the curvature as: $\forall\kX,\kY\in\Der_\eps^\bullet(\algA_\theta)$,
\begin{align}
A_\kX^g &=g\fois A_\kX\fois g^\ast+ig\fois \kX(g^\ast),\label{gtpot}\\
F_{\kX,\kY}^g &=g\fois F_{\kX,\kY}\fois g^\ast.\label{gtcurv}
\end{align}
\end{proposition}
\begin{proof}
$\forall \kX\in\Der_\eps^\bullet(\algA_\theta)$ and $\forall a\in\algrA_\theta$,
\begin{align*}
\nabla_\kX^g a=&g\fois (\kX(g^\ast\fois a)-iA_\kX\fois g^\ast\fois a)\\
=& \kX(a)+g\fois\kX(g^\ast)\fois a-ig\fois A_\kX\fois g^\ast\fois a.
\end{align*}
Since $\nabla^g_\kX a=\kX(a)-iA_\kX^g\fois a$, we obtain the result \eqref{gtpot}. The proof is similar for \eqref{gtcurv}.
\end{proof}

In the following, we will denote the gauge transformations (of degree 0) by $(g,0)$, where $g\in\caM$ and $g^\dag\star g=1$. Then, the gauge potentials \eqref{potmoy} transform as:
\begin{subequations}
\label{subeq-gtpot}
\begin{align}
(A_\mu^0,0)^g&=(g\star A_\mu^0\star g^\dag+ig\star\partial_\mu g^\dag,0),\label{gtpot1}\\
(0,A_\mu^1)^g&=(0,g\star A_\mu^1\star g^\dag+ig\star\partial_\mu g^\dag),\label{gtpot2}\\
(0,\varphi)^g&=(0,g\star\varphi\star g^\dag),\label{gtpot0}\\
(G_{\mu\nu},0)^g&=(g\star G_{\mu\nu}\star g^\dag-\frac i4g\star(\wx_\mu\partial_\nu g^\dag)-\frac i4g\star(\wx_\nu\partial_\mu g^\dag),0).\label{gtpot3}
\end{align}
\end{subequations}

We introduce the following covariant coordinates:
\begin{subequations}
\begin{align*}
\caA_\mu^0&=A_\mu^0+\frac 12\wx_\mu,&
\caA_\mu^1&=A_\mu^1+\frac 12\wx_\mu,\\
\Phi&=\varphi-1,&
\caG_{\mu\nu}&=G_{\mu\nu}-\frac 12\wx_\mu\wx_\nu,
\end{align*}
\end{subequations}
so that the curvature takes the new form:
\begin{subequations}
\label{subeq-moy2}
\begin{align}
F_{(0,i\gamma),(0,i\gamma)}&=(2i\alpha-2i\alpha\Phi\star\Phi,0)\\
F_{(i\xi_\mu,0),(0,i\gamma)}&=(0,-i[\caA_\mu^0,\Phi]_\star)\\
F_{(0,i\xi_\mu),(0,i\gamma)}&=(-i\alpha\{\caA_\mu^1,\Phi\}_\star-2i\alpha\caA_\mu^0,0)\\
F_{(i\eta_{\mu\nu},0),(0,i\gamma)}&=(0,-i[\caG_{\mu\nu},\Phi]_\star)\\
F_{(i\xi_\mu,0),(i\xi_\nu,0)}&=(\Theta^{-1}_{\mu\nu}-i[\caA_\mu^0,\caA_\nu^0]_\star,0)\\
F_{(i\xi_\mu,0),(0,i\xi_\nu)}&=(0,-i[\caA_\mu^0,\caA_\nu^1]_\star-\Theta^{-1}_{\mu\nu}\Phi)\\
F_{(0,i\xi_\mu),(0,i\xi_\nu)}&=(-i\alpha\{\caA_\mu^1,\caA_\nu^1\}_\star-i\alpha\caG_{\mu\nu},0)\\
F_{(i\xi_\mu,0),(i\eta_{\nu\rho},0)}&=(-i[\caA_\mu^0,\caG_{\nu\rho}]_\star+2\Theta^{-1}_{\nu\mu}\caA_\rho^0 +2\Theta^{-1}_{\rho\mu}\caA_\nu^0,0)\\
F_{(0,i\xi_\mu),(i\eta_{\nu\rho},0)}&=(0,-i[\caA_\mu^1,\caG_{\nu\rho}]_\star+2\Theta^{-1}_{\nu\mu}\caA_\rho^1 +2\Theta^{-1}_{\rho\mu}\caA_\nu^1)\\
F_{(i\eta_{\mu\nu},0),(i\eta_{\rho\sigma},0)}&=(-i[\caG_{\mu\nu},\caG_{\rho\sigma}]_\star-2\Theta^{-1}_{\mu\rho}\caG_{\nu\sigma} -2\Theta^{-1}_{\nu\rho}\caG_{\mu\sigma}-2\Theta^{-1}_{\mu\sigma}\caG_{\nu\rho}-2\Theta^{-1}_{\nu\sigma}\caG_{\mu\rho},0).
\end{align}
\end{subequations}

\medskip

In the notations of \eqref{subeq-moy1} and \eqref{subeq-moy2}, we will construct a gauge invariant action by taking the trace of the square of a (graded) curvature. Note that due to \eqref{subeq-gtpot}, the field $\varphi$ can be interpreted as a scalar field (in the adjoint representation), while $A_\mu^0$ and $A_\mu^1$ are usual gauge potentials on the Moyal space. Let us identify $A_\mu^0=A_\mu^1=A_\mu$ for the computation of the action. Note also that for dimensional consistency, one has to consider $i\sqrt{\theta}\eta_{\mu\nu}$ rather than $i\eta_{\mu\nu}$ in the action, $\frac{i}{\sqrt{\theta}}\gamma$ rather than $i\gamma$, and to rescale conveniently the fields $\varphi$ and $G_{\mu\nu}$. Taking for the moment $\caG_{\mu\nu}=0$ and $\Phi=0$, which can be done because these fields transform covariantly, we obtain the following action, up to constant terms:
\begin{multline}
\tr\Big(|F_{\ad_a,\ad_b}|^2\Big)= \\
\int d^Dx\Big((1+2\alpha)F_{\mu\nu}\star F_{\mu\nu}+\alpha^2\{\caA_\mu,\caA_\nu\}_\star^2+\frac{8}{\theta}(2(D+1)(1+\alpha)+\alpha^2)\caA_\mu\star\caA_\mu\Big),\label{actgaugeresult}
\end{multline}
where there is an implicit summation on $a,b\in\{(0,\frac{i}{\sqrt{\theta}}\gamma),(i\xi_\mu,0),(0,i\xi_\mu),(i\sqrt{\theta}\eta_{\mu\nu},0)\}$. We easily recognize the action \eqref{actgaugemoy}. When taking\footnote{We recall that $\Phi=\varphi-1$.} $\Phi\neq0$ and $\caG_{\mu\nu}=0$, we obtain an additional real scalar field action coupled to gauge fields:
\begin{multline}
S(\varphi)=\int d^Dx\Big(2\alpha|\partial_\mu\varphi-i[A_\mu,\varphi]_\star|^2+2\alpha^2|\wx_\mu\varphi+\{A_\mu,\varphi\}_\star|^2-4\alpha\sqrt{\theta}\varphi\Theta_{\mu\nu}^{-1}F_{\mu\nu}\\
+\frac{2\alpha(D+2\alpha)}{\theta}\varphi^2-\frac{8\alpha^2}{\sqrt{\theta}}\varphi\star\varphi\star\varphi+4\alpha^2\varphi\star\varphi\star\varphi\star\varphi \Big). \label{acthiggs}
\end{multline}
A part of this action can be interpreted as a Higgs action coupled to gauge fields, with harmonic term \eqref{actscalmoy} and a more general positive potential term built from the $\star$-polynomial part involving $\varphi$. Moreover, this action is not Langmann-Szabo covariant (cubic term), but related to the generalization of this duality discussed below.

Notice that the gauge fields $A_\mu$ are already massive (see the last term in \eqref{actgaugeresult}). Furthermore, the action \eqref{acthiggs} involves a BF-like \cite{Wallet:1989wr,Birmingham:1991} term $\int \varphi\Theta^{-1}_{\mu\nu}F_{\mu\nu}$ similar to the one introduced by Slavnov \cite{Slavnov:2003ae} in the simplest noncommutative extension of the Yang-Mills theory on Moyal space.

\subsubsection{Discussion on the gauge theory}

We have therefore interpreted both actions \eqref{actscalmoy} and \eqref{actgaugemoy} within the formalism introduced in this subsection for the superalgebra $\algrA_\theta$. Some comments concerning the classical action \eqref{actgaugemoy} are now in order.

First, it is not surprising that we obtain a Higgs field in this theory as a part of a connection. This was already the case e.g. for spectral triples approach of the Standard Model \cite{Chamseddine:2006ep} and for the models stemming from the derivation-based differential calculus of \cite{DuboisViolette:1989vq}. Notice that a somewhat similar interpretation of covariant coordinates as Higgs fields in the context of gauge theory models on Moyal algebras has been also given in \cite{Cagnache:2008tz}. Moreover, an additional BF-like term in \eqref{acthiggs} appears in the present situation. Such a Slavnov term has been shown \cite{Slavnov:2003ae}  to improve the IR (and UV) behavior of the following action
\begin{equation}
\int d^Dx\Big((1+2\alpha)F_{\mu\nu}\star F_{\mu\nu}-4\alpha\sqrt{\theta}\varphi\Theta_{\mu\nu}^{-1}F_{\mu\nu}\Big), \label{actBF}
\end{equation}
when the field $\varphi$ is not dynamical \cite{Slavnov:2003ae}. The corresponding impact on the UV/IR mixing of the full action in the present situation remains to be analyzed.

In view of the discussion given in part \ref{subsub-moyalalg}, the grading of the $\eps$-associative algebra $\algrA$ mimics the Langmann-Szabo duality in the scalar case. As far as the gauge theory built from the square of the curvature is concerned, this is reflected in particular into the action \eqref{actgaugeresult} (The $\caG_{\mu\nu}=\Phi=0$ part), which has been found as an effective action 
\cite{deGoursac:2007gq, Grosse:2007dm}. As a consequence, one can view this grading as a generalization of the Langmann-Szabo duality to the scalar and gauge case. Observe that in \eqref{actgaugeresult}, the part $(\{\caA_\mu,\caA_\nu\}_\star-\frac14\{\wx_\mu,\wx_\nu\}_\star)$ might be viewed as the symmetric counterpart of the usual antisymmetric curvature $F_{\mu\nu}=\frac i4[\wx_\mu,\wx_\nu]_\star-i[\caA_\mu,\caA_\nu]_\star$. Then, the general action $\tr\Big(|F_{\ad_a,\ad_b}|^2\Big)$ involves the terms
\begin{equation}
\int d^Dx\Big(\alpha^2(\{\caA_\mu,\caA_\nu\}_\star-\frac 12\wx_\mu\wx_\nu)^2+ 2\alpha^2G_{\mu\nu}(\{\caA_\mu,\caA_\nu\}_\star-\frac 12\wx_\mu\wx_\nu)\Big),%\label{actBFsym}
\end{equation}
which can also be seen as the symmetric counterpart of the BF action \eqref{actBF}, if we interpret $G_{\mu\nu}$ as the ``symmetric'' analog of the BF multiplier $\Theta_{\mu\nu}^{-1}\varphi$.

One can observe further that the gauge action built from the square of the curvature has a trivial vacuum (and in particular $A_\mu=0$) which avoid therefore the difficult problem \cite{deGoursac:2007qi} to deal with the non trivial vacuum configurations of \eqref{actgaugeresult} that have been determined in \cite{deGoursac:2008rb}.

Thanks to the formalism of derivation-based differential calculus and $\eps$-connections developed in section \ref{sec-ncg}, we have thus given a mathematical interpretation of the theory \eqref{actgaugemoy} as constructed from a graded curvature.

\vskip 1 true cm

\bibliographystyle{utcaps}
\bibliography{biblio-these,biblio-perso,biblio-recents}

\end{document}